%% file: main.tex
\documentclass[11pt]{article}
\usepackage{amsmath}
\usepackage{amsthm}
\usepackage{amssymb}
\usepackage{algorithm}
\usepackage{subfig}
\usepackage{color}
\usepackage[english]{babel}
\usepackage{graphicx}
\usepackage{epstopdf}
\usepackage{url}
\usepackage{graphicx}
\usepackage{color}
\usepackage{algpseudocode}
\usepackage{bbm}
\usepackage{comment}
\usepackage{authblk}
\usepackage{dsfont}
\usepackage{multicol}
\usepackage{enumitem}

\let\C\relax
\usepackage{tikz}
\usepackage{hyperref}  
\usepackage[margin=1in]{geometry}
\hypersetup{colorlinks=true,citecolor=blue,linkcolor=blue} 
\usetikzlibrary{arrows}
\graphicspath{{./figs/}}

\author[1]{Yi Li}
\author[2]{Vasileios Nakos\thanks{This work is part of the project TIPEA that has received funding from the European
Research Council (ERC) under the European Unions Horizon 2020 research and innovation programme (grant agreement No. 850979). Part of the work was completed when the author was a Ph.D. student in Harvard University and supported in part by NSF CAREER award CCF-1350670.}}
\affil[1]{Nanyang Technological University\\ \texttt{yili@ntu.edu.sg}}
\affil[2]{Saarland University and Max-Planck Institute for Informatics\\ \texttt{vnakos@mpi-inf.mpg.de}}

\date{}

\title{Deterministic Sparse Fourier Transform with an $\ell_{\infty}$ Guarantee}
\newtheorem{theorem}{Theorem}[section]
\newtheorem{lemma}[theorem]{Lemma}
\newtheorem{definition}[theorem]{Definition}

\newtheorem{proposition}[theorem]{Proposition}

\newtheorem{remark}[theorem]{Remark}

\newtheorem{question}[theorem]{Question}

\newcommand{\wh}{\widehat}
\newcommand{\wt}{\widetilde}

\newcommand{\eps}{\epsilon}

\renewcommand{\varepsilon}{\epsilon}
\renewcommand{\tilde}{\wt}
\renewcommand{\hat}{\wh}

\newcommand{\I}{\sqrt{-1}}

\renewcommand{\Re}{\operatorname{Re}}
\renewcommand{\Im}{\operatorname{Im}}

\DeclareMathOperator*{\E}{{\bf {E}}}

\DeclareMathOperator{\Z}{\mathbb{Z}}
\DeclareMathOperator{\R}{\mathbb{R}}
\DeclareMathOperator{\C}{\mathbb{C}}

\DeclareMathOperator*{\median}{median}

\DeclareMathOperator{\supp}{supp}

\DeclareMathOperator{\poly}{poly}

\newcommand{\Yi}[1]{\textbf{\color{blue}[Yi: #1]}}

\makeatletter
\newcommand*{\RN}[1]{\expandafter\@slowromancap\romannumeral #1@}
\makeatother

\makeatletter
\newtheorem*{rep@theorem}{\rep@title}
\newcommand{\newreptheorem}[2]{%
\newenvironment{rep#1}[1]{%
 \def\rep@title{#2 \ref{##1}}%
 \begin{rep@theorem}}%
 {\end{rep@theorem}}}
\makeatother

\newreptheorem{theorem}{Theorem}
\newreptheorem{lemma}{Lemma}

\usepackage{lineno}

\input{libtheorems.tex}

\begin{document}

\begin{titlepage}
  \maketitle
  \begin{abstract}
\input{abstract}

  \end{abstract}
  \thispagestyle{empty}
\end{titlepage}




\input{introduction}

\input{intro}
\input{Overview}
\input{preliminaries}

\input{linear_algo}

\input{sublinear_algo}

\input{incoherence}
\input{weil}
\input{open_problems}





\newpage
\addcontentsline{toc}{section}{References}
\bibliographystyle{alpha}
\bibliography{ref}
\newpage

\appendix

\input{ell_inf_reduction}




\end{document}

%% file: libtheorems.tex
\usepackage{etoolbox}

\makeatletter

\newcommand{\define}[4][ignore]{%
  \ifstrequal{#1}{ignore}{}{
  \@namedef{thmtitle@#2}{#1}}%
  \@namedef{thm@#2}{#4}%
  \@namedef{thmtypen@#2}{lemma}%
  \newtheorem{thmtype@#2}[theorem]{#3}%
  \newtheorem*{thmtypealt@#2}{#3~\ref{#2}}%
}

\newcommand{\state}[1]{%
  \@namedef{curthm}{#1}
  \@ifundefined{thmtitle@#1}{
  \begin{thmtype@#1}
    }{
  \begin{thmtype@#1}[\@nameuse{thmtitle@#1}]
  }
    \label{#1}
    \@nameuse{thm@#1}
  \end{thmtype@#1}
  \@ifundefined{thmdone@#1}{
  \@namedef{thmdone@#1}{stated}%
  }{}
}

\newcommand{\restate}[1]{%
  \@namedef{curthm}{#1}
  \@ifundefined{thmtitle@#1}{
    \begin{thmtypealt@#1}
    }{
  \begin{thmtypealt@#1}[\@nameuse{thmtitle@#1}]
  }
    \@nameuse{thm@#1}
  \end{thmtypealt@#1}
  \@ifundefined{thmdone@#1}{
  \@namedef{thmdone@#1}{stated}%
  }{}
}

\newcommand{\thmlabel}[1]{
  \@ifundefined{thmdone@\@nameuse{curthm}}{\label{#1}
    }{\tag*{\eqref{#1}}}
}
\makeatother

%% file: abstract.tex

In this paper we revisit the deterministic version of the Sparse Fourier Transform problem, which asks to read only a few entries of $x \in \mathbb{C}^n$ and design a recovery algorithm such that the output of the algorithm approximates $\hat x$, the Discrete Fourier Transform (DFT) of $x$. The randomized case has been well-understood, while the main work in the deterministic case is that of Merhi et al.\@ (J Fourier Anal Appl 2018), which obtains $O(k^2 \log^{-1}k \cdot \log^{5.5}n)$ samples and a similar runtime with the $\ell_2/\ell_1$ guarantee. We focus on the stronger $\ell_{\infty}/\ell_1$ guarantee and the closely related problem of incoherent matrices. We list our contributions as follows.

\begin{enumerate}
\item We find a deterministic collection of $O(k^2 \log n)$ samples for the $\ell_\infty/\ell_1$ recovery in time $O(nk \log^2 n)$, and a deterministic collection of $O(k^2 \log^2 n)$ samples for the $\ell_\infty/\ell_1$ sparse recovery in time $O(k^2 \log^3n)$.
\item We give new deterministic constructions of incoherent matrices that are row-sampled submatrices of the DFT matrix, via a derandomization of Bernstein's inequality and bounds on exponential sums considered in analytic number theory.
Our first construction matches a previous randomized construction of Nelson, Nguyen and Woodruff (RANDOM 12), where there was no constraint on the form of the incoherent matrix. 
\end{enumerate}

Our algorithms are nearly sample-optimal, since a lower bound of $\Omega(k^2 + k \log n)$ is known, even for the case where the sensing matrix can be arbitrarily designed. A similar lower bound of $\Omega(k^2 \log n/ \log k)$ is known for incoherent matrices.

%% file: introduction.tex
\section{Introduction}

Compressed sensing is a subfield of discrete signal processing, based on the principle that a high-dimensional signal can be approximately reconstructed, by exploiting its sparsity, in fewer samples than those demanded by the Shannon-Nyquist theorem. An important subtopic is the Sparse Fourier Transform, where we desire to detect and approximate the largest coordinates of a high-dimensional signal, given a few samples from its Fourier spectrum. Fewer samples play a crucial role, for example, in medical imaging, where reconstructing an image corresponds exactly to reconstructing a signal from its Fourier representation. Thus, the number of Fourier coefficients needed for (approximate) reconstruction is proportional to the radiation dose a patient receives as well as the time the patient needs to remain in the scanner. Furthermore, exploiting the sparsity of the signal has given researchers the hope of defeating the FFT algorithm of Cooley and Tukey, in the special (but of high practical value) case where the signal is approximately sparse. Thus, since FFT serves as an important computational primitive, and has been recognized as one of the 10 most important algorithms of the 20th century \cite{cipra2000best}, every place where it has found application can possibly be benefited from a faster algorithm. The main intuition and hope is that signals arising in practice often exhibit certain structure, such as concentration of energy in a small number of Fourier coefficients.

Since vectors in practice are never exactly sparse, and it is impossible to reconstruct a generic vector $\wh{x} \in \mathbb{C}^n$ from $o(n)$ samples, researchers resort to approximation. More formally, a sparse recovery scheme consists of a sample set $S\subseteq \{1,\dots,n\}$ and a recovery algorithm $\mathcal{R}$ such that for any given $x\in \C^n$, the scheme approximates $\hat x$ by $\hat x' = \mathcal{R}(x_S)$, where $x_S$ denotes the vector of $x$ restricted to the coordinates in $S$. The fineness of approximation is measured with respect to the best $k$-sparse approximation to $\wh{x}$. The breakthrough work of Cand\`es, Tao and Donoho~\cite{ct06,d06} first showed that $k\log^{O(1)}n$ samples of $x \in \mathbb{C}^n$ suffices to reconstruct a $O(k)$-sparse vector $\wh{x}'$ which is ``close'' to the best $k$-approximation of $\wh{x}$. More formally, the reconstruction $\wh{x}'$ satisfies the so-called  $\ell_2/\ell_1$ guarantee, i.e.,
\[
\|\wh{x} - \wh{x}'\|_2 \leq \frac{1}{\sqrt{k}} \| \wh{x}_{-k}\|_1,
\]
where $\wh{x}_{-k}$ is the tail vector, obtained from restricting $\wh{x}$ to its smallest $n-k$ coordinates in magnitude. The strength of their algorithms lies in the uniformity, in the sense that the samples at the same coordinates can be used to approximate every $x \in \mathbb{C}^n$. However, the running time is polynomial in the vector length $n$, giving thus only sample-efficient, but not necessarily time-efficient, algorithms. Furthermore, the samples are not obtained via a deterministic procedure, but are chosen at random. Regarding non-uniform randomized algorithms that run in sublinear time, numerous researchers have worked on the problem and obtained a series of algorithms with different recovery guarantees \cite{gl89,m92,km93,ggims02,akaviagoldwasser03,gms05,hikp12a,hikp12b,lawlor13,iw13,PR14,ikp14,ik14,k16,k17,kvz19,nakos2019nearly}.
See Table~\ref{tab:different_guarantees} for a list of common recovery guarantees.
 The state of the art is the seminal algorithm of Kapralov \cite{k17}, which shows that $O(k \log n)$ samples and $O(k \log^{O(1)}n)$ time are simultaneously possible for the $\ell_2/\ell_2$ guarantee (which is strictly stronger\footnote{Here we mean that given an algorithm giving the $\ell_2/\ell_2$ guarantee, one can create an algorithm, using the $\ell_2/\ell_2$ algorithm as a black box,  with sparsity parameter $k' = O(k)$, achieving the $\ell_2/\ell_1$ guarantee with the same order of number of samples.} than the $\ell_2/\ell_1$). The fastest algorithm is due to \cite{hikp12a}, needing $O(k \log n \cdot \log(n/k))$ time and samples. We note also the algorithm of Indyk and Kapralov \cite{ik14} that runs in $O(n \log^2 n)$ time, uses $O(k \log n)$ samples but gives a stronger $\ell_{\infty}/\ell_2$ guarantee than the $\ell_2/\ell_2$ guarantee in the previous two papers. We refer the reader to the next section for comparison of the different guarantees appearing in the literature. Recently there has been also considerable work on recovering $k$-sparse signals from their continuous Fourier Transform, see \cite{bcgls14,ps15,ckps16,hkmmvz19}.

\begin{table}
\begin{center}
    \begin{tabular}{| l | l | l |}
    \hline
       Guarantee & Formula &  Deterministic Lower Bound   \\ \hline
    	$\ell_{\infty}/\ell_2$ & $\|\hat{x}-\hat{x}'\|_{\infty} \leq \|\hat{x}_{-k}\|_2/\sqrt{k}$ & $\Omega(n)$ \cite{cdd09}\\ \hline
    
    $\ell_{2}/\ell_2$ & $\|\hat{x}-\hat{x}'\|_{2} \leq C\|\hat{x}_{-k}\|_2$ & $\Omega(n)$ \cite{cdd09} \\ \hline

    $\ell_{\infty}/\ell_1$ & $\|\hat{x}-\hat{x}'\|_{\infty} \leq \|\hat{x}_{-k}\|_1/k$ & $\Omega(k^2 + k \log n)$ \cite{ganguly2008lower,fpru10} \\ \hline
    $\ell_{2}/\ell_1$ & $\|\hat{x}-\hat{x}'\|_2 \leq \|\hat{x}_{-k}\|_1/\sqrt{k}$ & $\Omega(k \log (n/k))$ \cite{ganguly2008lower,fpru10}\\ \hline	
\end{tabular}
\end{center}
\caption{Common guarantees of sparse recovery. Only the $\ell_2/\ell_2$ case requires a parameter $C>1$. The guarantees are listed in the descending order of strength. }
\label{tab:different_guarantees}
\end{table}

Although our understanding on randomized algorithms is almost complete, there are still important gaps in our knowledge regarding deterministic schemes. The following natural open-ended question has theoretical and practical interest and remains in principle highly unexplored, touching a variety of fields including (sublinear-time) algorithms, pseudorandomness and computational complexity, Additive Combinatorics \cite{bourgain2011breaking} and analytic number theory.

\begin{question}
What are the best bounds we can obtain for the different versions of the deterministic Sparse Fourier Transform problem? 
\end{question}

With sublinear runtime, the earliest work of Iwen \cite{i08,i10} gives $O(k^2 \log^4 n)$ samples and time, albeit in a significantly easier (although similar) model: where one wants to learn a band-limited function $f : [0, 2\pi) \to \mathbb{C}$ and can evaluate $f$ at any point. In the discrete case which we are interested in, the state of the art is the work of Merhi et al.~\cite{mzic2017}, which obtains $O(k^2 \log^{11/2} n/\log k)$ samples and the same runtime. A recent work of Bittens et al.~\cite{bzi17} showed that the quadratic dependence can be dropped if the signals are sufficiently structured, namely, if the Fourier coefficients are generated by an unknown but small degree polynomial. On the related problem of the Walsh-Hamadard Transform, Indyk and Cheraghchi \cite{ci17} showed that roughly $O(k^{1+\alpha} \log^{O(1) + 6/\alpha} n)$ samples and similar run-time are possible, if one resorts to a slightly weaker guarantee. Interestingly, their approach resides in a novel connection between the Walsh-Hadamard matrix and linear lossless condensers. However, this connection does not extend to the Fourier Transform over $\mathbb{Z}_n$, which is our focus and the most interesting case. Interesting ideas appear also in the work of Akavia~\cite{akavia10,akavia14}, where it is shown how to approximate the Fourier Transform of an arithmetic progression in poly-logarithmic time in the length of the progression; due to the worse dependence on the quality of approximation, however, that work obtained an algorithm with sample complexity $(k\cdot \text{(signal-to-noise ratio)})^4$.

The papers above showed how to achieve the $\ell_2/\ell_1$ guarantee in a number of samples that is quadratic in the signal sparsity. It is already known that a nearly linear dependence is possible~\cite{ct06}; however, we do not have efficient deterministic algorithms for finding these samples. The work of \cite{ct06}, as well as subsequent works, proceeds by sampling with repetition rows of the DFT matrix, and showing that the RIP condition (see Definition~\ref{def:RIP}) holds, which in turn implies the desired result, but via a super-linear algorithm. The state-of-the-art analysis of such row subsampling is due to Haviv and Regev~\cite{hr16}, who showed that $O(k\log^2 k \log n)$ samples suffice. A lower bound of $\Omega(k \log n)$ rows for this subsampling process has been shown in \cite{bld2017}. In this paper, we follow a different avenue and give a new set of schemes for the Sparse Fourier Transform which allow uniform reconstruction. Although our dependence is still quadratic in $k$, it is necessary, in contrast to the previous works: our results satisfy the strictly stronger $\ell_{\infty}/\ell_1$ guarantee, for which a quadratic lower bound is known \cite{ganguly2008lower}, and hence one cannot hope for a sub-quadratic dependence. We also note the deterministic algorithm of \cite{kvz19}, which needs a cubic dependence on $k$ but solves a somewhat different problem of finding the multidimensional sparse Fourier transform of a signal with at most $k$ non-zeros in the frequency domain, and thus is not robust to noise.

The focus of our work is the $\ell_\infty/\ell_1$ guarantee, defined formally as follows.

\begin{definition}[$\ell_{\infty}/\ell_1$ guarantee]\label{def:ell_inf}
A sparse recovery scheme is said to satisfy the $\ell_{\infty}/\ell_1$ guarantee with parameter $k$, if given access to vector $x$, it outputs a vector $\hat{x}'$ such that 
\begin{equation}\label{eqn:approx_guarantee}
	\|\hat{x} - \hat{x}'\|_{\infty} \leq \frac{1}{k} \|\hat{x}_{-k}\|_1.
\end{equation}
\end{definition}

\paragraph{$\ell_\infty/\ell_1$ versus $\ell_2/\ell_1$: A matter of ``find all'' versus ``miss all''.} As we have discussed, previous works satisfied the $\ell_2/\ell_1$ guarantee, while our target is the $\ell_\infty/\ell_1$ guarantee. Any algorithm for the latter guarantee also satisfies the former one. But, as we shall demonstrate in Section~\ref{sec:comparison}, the $\ell_\infty/\ell_1$ guarantee is much stronger: there exists an infinite family of vectors for which an  $\ell_2/\ell_1$ algorithm might detect none of the heavy frequencies, while an $\ell_\infty/\ell_1$ algorithm must detect all of them. This happens because the $\ell_\infty/\ell_1$ is a \textbf{worst-case} guarantee, in the sense that it requires detection of every frequency just above the noise level, in contrast to the $\ell_2/\ell_1$, which should be regarded as an \textbf{average-case} guarantee in the sense that it allows missing a subset of the heavy frequencies if they carry the energy proportional to the noise level.

\paragraph{Previous Work on $\ell_\infty/\ell_1$ with arbitrary linear measurements.}
All approaches described above concerned Fourier measurements, but compressed sensing has a long history using arbitrary linear measurements, for example \cite{dipw10,pw11,ipw11,glps12,gnprs13,nakos2018adaptive,glps17,lnw17,nakos2017deterministic_heavy_hitters,nakos2019stronger}. Regarding $\ell_{\infty}/\ell_1$, the work of \cite{nnw14} indicated a connection between the aforementioned guarantee and incoherent matrices. More specifically, it was shown that given a $(1/k)$-incoherent matrix one can design an algorithm satisfying the $\ell_{\infty}/\ell_1$ guarantee. The existence of a matrix with $O(k^2 \min \{ \log n , (\log n/ \log k)^2\})$ rows was also proved. Reconstruction 
needed $\Omega(n k)$ time, something which was partially remedied by Li and Nakos~\cite{nakos2017deterministic_heavy_hitters} with a scheme of $O(k^2\log n \cdot \log^{\ast}k)$ measurements and $\poly(k, \log n)$ decoding time. Incoherent matrices are interesting objects on their own, and have been studied before, as they can be used to obtain RIP matrices. Deterministic constructions of $O(k^2 (\log n/ \log k)^2)$ rows were obtained by DeVore \cite{devore2007} using deep results from the theory of Gelfand widths and by Amini and Marvasti~\cite{amini2011deterministic} via binary BCH code vectors, where the zeros are replaced by $-1$s. We note that incoherent matrices matching this bound also follow immediately from the famous Nisan-Wigderson combinatorial designs \cite{nisan1994hardness}, and serve as a cornerstone for constructions of pseudorandom generators and extractors~\cite{trevisan2001extractors}. Incoherent matrices are also connected with $\epsilon$-biased codes, and thus an almost optimal strongly explicit construction can be obtained by the recent breakthrough work of \cite{TaShma}. On the lower bound side, Alon has shown that $\Omega(k^2 \log n / \log k)$ rows are necessary for a $(1/k)$-incoherent matrix~\cite{alon2009perturbed}. 

\paragraph{Our Contribution.} In this work we offer several new results for the Sparse Fourier Transform problem across different axis, some of which are nearly optimal. We show how to find in polynomial time a deterministic collection of samples from the time domain, such that we can solve the Sparse Fourier Transform problem  in linear and sublinear time and achieve nearly optimal sample complexity. 
For the closely related problem of incoherent matrices from DFT rows, which is of independent interest, we obtain a nearly optimal derandomized construction via Bernstein's inequality. We also demonstrate strongly explicit constructions, by invoking heavy number-theoretical machinery. 

We note that the bounds of our constructions have been known for more than a decade if the sensing/incoherent matrix is allowed to be arbitrary. However, the previous arguments did not facilitate the frequent and relevant scenario where we have access to rows only from the Fourier ensemble. Part of our work is to show that some of these results carry over to the significantly more constrained case. We also note that any progress to deterministic $\ell_2/\ell_1$ schemes with subquadratic sample complexity is connected to the very challenging problem of obtaining a deterministic DFT row-subsampled RIP matrices with subquadratic number of rows\footnote{Note that \cite{bourgain2011breaking} breaks the quadratic barrier for RIP matrices but does not use the Fourier ensemble; the rows are picked from the \emph{discrete chirp-Fourier} ensemble, where the linear functions are substituted by quadratic polynomials.} which possibly out of reach at the moment.

%% file: intro.tex
\section{Technical Results}

\subsection{Preliminaries}

For a positive integer $n$, we define $[n] = \{0,1\ldots,n-1\}$ and we shall index the coordinates of a $n$-dimensional vector or the rows/columns of an $n\times n$ matrix from $0$ to $n-1$. We define the Discrete Fourier Transform (DFT) matrix $F \in \mathbb{C}^{n \times n}$ to be the unitary matrix such that 
$F_{ij} = \frac{1}{\sqrt{n}}e^{2\pi\sqrt{-1}\cdot ij/n}$, and the Discrete Fourier Transform of a vector $x\in\C^n$ to be $\hat{x} = Fx$.  

For a set $S\subseteq [n]$ we define $x_S$ to be the vector obtained from $x$ after zeroing out the coordinates not in $S$. We also define $H(x,k)$ to be the set of the indices of the largest $k$ coordinates (in magnitude) of $x$, and $x_{-k} = x_{ [n] \setminus H(x,k)}$. We say $x$ is $k$-sparse if $x_{-k} = 0$. We also define $\|x\|_p = \big( \sum_{i=0}^{n-1} |x_i|^p \big)^{1/p}$ for $p \geq 1$ and $\|x\|_0$ to be the number of nonzero coordinates of $x$.

For a matrix $F \in \C^{n \times n}$ and subsets $S,T \subseteq [n]$, we define $F_{S,T}$ to be the submatrix of $F$ indexed by rows in $S$ and columns in $T$.

The median of a collection of complex numbers $\{z_i\}$ is defined to be $\median_i z_i = \median_i \Re(z_i) + \sqrt{-1}\median_i \Im(z_i)$,
i.e., taking the median of the real and the imaginary component separately.

For two points $x$ and $y$ on the unit circle, we use $|x-y|_\circ$ to denote the circular distance (in radians, i.e. modulo $2\pi$) between $x$ and $y$.

\subsubsection{$\ell_\infty/\ell_1$ Gurantee and incoherent matrices}
 The quality of the approximation is usually measured in different error metrics, and the main recovery guarantee we are interested in is called the 
$\ell_\infty/\ell_1$ guarantee, as defined in Definition \ref{def:ell_inf}.
%
%
Other types of recovery guarantee, such as the $\ell_\infty/\ell_2$, the $\ell_2/\ell_2$ and the $\ell_2/\ell_1$, are defined similarly, where \eqref{eqn:approx_guarantee} is replaced with the respective expression in Table~\ref{tab:different_guarantees}. Note that these are definitions of the error guarantee per se and do not have algorithmic requirements on the scheme.

Highly relevant with the $\ell_{\infty}/\ell_1$ guarantee is a matrix condition which we call incoherence.
\begin{definition}[Incoherent Matrix]
A matrix $A \in \mathbb{C}^{m\times n}$ is called $\epsilon$-incoherent if $\|A_i\|_2 =1$ for all $i$ (where $A_i$ denotes the $i$-th column of $A$) and $|\langle A_i,A_j \rangle| \leq \epsilon$.
\end{definition}

\begin{lemma}[\cite{nnw14}] 
There exist an absolute constant $c>0$ such that for any $(c/k)$-incoherent matrix $A$, there exists a $\ell_{\infty}/\ell_1$-scheme which uses $A$ as the measurement matrix and whose recovery algorithm runs in polynomial time. 
\end{lemma}

\subsubsection{The Restrictred Isometry Property and its connection with incoherence}
Another highly relevant condition is called the renowned restricted isometry property, introduced by Cand\`es et al.\@ in~\cite{crt06}. We show how incoherent matrices are connected to it.

\begin{definition}[Restricted Isometry Property]\label{def:RIP}
A matrix $A \in \mathbb{C}^{m\times n}$ is said to satisfy the $(k,\epsilon)$ Restricted Isometry Property (RIP), if for all $x \in \mathbb{C}^n$ with $\|x\|_0\leq k$, it holds that $(1-\epsilon) \|x\|_2 \leq \|Ax\|_2 \leq (1+\epsilon) \|x\|_2$.
\end{definition}

Cand\`es et al.\@ proved in their breakthrough paper \cite{crt06} that any RIP matrix can be used for sparse recovery with the $\ell_2/\ell_1$ error guarantee. The following formulation comes from~\cite[Theorem 6.12]{FR}.
\begin{lemma} \label{thm:rip_to_sparserec}
Given a $(2k,\epsilon)$-RIP matrix $A$ with $\epsilon < 4/\sqrt{41}$, we can design a $\ell_2/\ell_1$-scheme that uses $A$ as the measurement matrix and has a recovery algorithm that runs in polynomial time.
\end{lemma}

Although randomly subsampling the DFT matrix gives an RIP matrix with $O(k \log^2 k \log n)$ rows \cite{hr16}, no algorithm for finding these rows in polynomial time is known; actually, even for $o(k^2) \cdot \poly( \log n)$ rows the problem remains wide open\footnote{In fact, one of the results of our paper gives the state-of-the-art result even for this problem, with $O(k^2 \log n)$ rows, see Theorem \ref{thm:incoherent1}.}. It is a very important and challenging problem whether one can have an explicit construction of RIP matrices from Fourier measurements that break the quadratic barrier on $k$. 

We state the following two folklore results, connecting the two different guarantees, and their associated combinatorial objects. This indicates the importance of incoherent matrices for the field of compressed sensing.

\begin{proposition}[folklore]
An $\ell_{\infty}/\ell_1$ scheme with a measurement matrix of $m$ rows and recovery time $T$ induces an $\ell_2/\ell_1$ scheme of a measurement matrix of $O(m)$ rows and recovery time $O(T + \|\hat x'\|_0)$, where $\hat x'$ is the output of the $\ell_{\infty}/\ell_1$ scheme.
\end{proposition}

\begin{proposition}[folklore]
A $(c/k)$-incoherent matrix is also a $(k,c)$-RIP matrix.
\end{proposition}


\subsection{Our results}

\subsubsection{Sparse Fourier Transform Algorithms}\label{subsec:sft}

\begin{table} 
\begin{center}
    \begin{tabular}{| l | l | l | l |l |l |}
    \hline
     & Samples& Run-time & Guarantee & \shortstack{Explict\\ Construction} & Lower Bound \\ \hline
    	\cite{hr16}  & $k \log^2k \log n$ & $\poly(n)$ & $\ell_2/\ell_1$ &  No & $k \log (n/k)$  \\ \hline
	\cite{mzic2017} & $k^2 \log^{5.5} n / \log k$ & $k^2 \log^{5.5} n/\log k$ & $\ell_2/\ell_1$ & Yes & $k\log(n/k)$  \\ \hline
   Theorem~\ref{thm:intro_main_1} & $k^2 \log n$ & $n k \log^2 n$ & $\ell_{\infty}/\ell_1$ & Yes& $k^2 + k \log n \cite{nnw14}$\\ \hline
   Theorem~\ref{thm:intro_main_2} & $k^2 \log^2 n$ & $k^2 \log^3 n$ & $\ell_{\infty}/\ell_1$ & Yes & $k^2 + k \log n \cite{nnw14}$\\ \hline
    \end{tabular}
\end{center}
\caption{Comparison of our results and the previous results. All $O$- and $\Omega$-notations are suppressed. The result in the first row follows from Lemma~\ref{thm:rip_to_sparserec} and the RIP matrix in~\cite{hr16}.Our algorithms adopt the common assumption in the sparse FT literature that the signal-to-noise ratio is bounded by $n^c$ for some absolute constant $c > 0$.} 
\label{tab:Fourier_results}
\end{table}

\begin{theorem}[Deterministic SFT with super-linear time, Section~\ref{sec:linear_time}]\label{thm:intro_main_1}
Let $n$ be a power of $2$. There exist a set $S \subseteq [n]$ with $|S| = O(k^2 \log n)$ and an absolute constant $c>0$ such that the following holds. For any vector $x \in \mathbb{C}^n$ with $\|\hat{x}\|_{\infty} \leq n^c \|\hat{x}_{-k}\|_1/k$, one can find an $O(k)$-sparse vector $\hat{x}' \in \C^n $ such that
\[
\| \hat{x} - \hat{x}' \|_{\infty} \leq \frac{1}{k} \|\hat{x}_{-k}\|_1, 
\]
in time $O(n k \log^2n)$ by accessing $\{x_i\}_{i \in S}$ only. Moreover, the set $S$ can be found in $\poly(n)$ time.
\end{theorem}

\begin{theorem}[Deterministic SFT with sublinear time, Section~\ref{sec:sublinear_time}]\label{thm:intro_main_2}
Let $n$ be a power of $2$. There exist a set $S \subseteq [n]$ with $|S| = O(k^2 \log^2 n)$ and an absolute constant $c>0$ such that the following holds. For any vector $x \in \mathbb{C}^n$ with $\|\hat{x}\|_{\infty} \leq n^c \|\hat{x}_{-k}\|_1/k$, one can find an $O(k)$-sparse vector $\hat{x}' \in \C^n $ such that
\[
\| \hat{x} - \hat{x}'\|_{\infty} \leq \frac{1}{k} \|\hat{x}_{-k}\|_1 ,
\]
in time $O(k^2 \log^3 n)$ by accessing $\{x_i\}_{i \in S}$ only. Moreover, the set $S$ can be found in $\poly(n)$ time.
\end{theorem}

\begin{remark}
The condition $\|\hat{x}\|_{\infty} \leq n^c \|\hat{x}_{-k}\|_1/k$ upper bounds the ``signal-to-noise ratio'', a common measure in engineering that compares the level of a desired signal to the level of the background noise. This is a common assumption in most algorithms in the Sparse Fourier Transform literature, see, e.g.~\cite{hikp12a,ik14,k16,cksz17,k17}, where the $\ell_2$-norm variant $\|\hat{x}\|_\infty \leq n^c \|\hat{x}_{-k}\|_2/\sqrt{k}$ was assumed.
\end{remark}

\subsubsection{From DFT to incoherent matrices}\label{subsec:incoherent}

This section contains deterministic constructions of incoherent matrices. 
\paragraph{An Explicit Construction: Derandomization in $\mathrm{poly}(n)$ time.}
\begin{theorem}[Incoherent matrices by derandomized subsampling of DFT, Section~\ref{sec:incoherent_appendix}] \label{thm:incoherent1}
There exists a set $S \subseteq [n]$ with of cardinality $O(k^2 \log n)$ such that the matrix $\sqrt{\frac{n}{m}} F_{S,[n]}$ is $(1/k)$-incoherent. Moreover, $S$ can be found in $\poly(n)$ time.
\end{theorem}

The above Theorem yields immediately a different algorithm for $\ell_\infty/\ell_1$ Sparse Fourier Tranform with $O(k^2\log n)$ samples, via the reduction in~\cite{nnw14}.

\paragraph{Strongly explicit constructions: Derandomization in sub-linear time}

\begin{theorem}[Incoherent matrices from DFT via low-degree polynomials, Section~\ref{sec:weil}] \label{thm:fourier_strongly_explicit}
Let $\epsilon>0$ be a constant small enough, $p$ be a prime and $d \geq 2$ be an integer. There exists a strongly explicit construction of an $O(m^\epsilon(\frac{1}{m} + \frac{p}{m^d})^{2^{1-d}})$-incoherent matrix $M\in \C^{m\times p}$ such that the rows of $\sqrt{m}M$ are rows of the DFT matrix (a row may appear more than once). The hidden constant in the $O$-notation depends on $d$ and $\eps$. Finding the indices of the rows takes $\tilde{O}(m)$ time.
\end{theorem}

To get an idea of the above result one could for example set $d=3$ and observe that the results translates to the following: for every $k \geq p^{1/8}$ one can get a $(1/k)$-incoherent matrix with $O(k^{4+\eps})$ rows. One needs the condition on $k$ (or equivalently the condition on $m$) to bound the term $p/m^d$. The larger the degree $d$, the looser this condition, but also the worse the dependence of $m$ on $k$. For example, when $d=4$, we can expand the regime of $k$ to  approximately $k \geq p^{1/24}$, but obtain approximately $m=O(k^{8+\eps})$. 

The following is a different construction, incomparable with Theorem~\ref{thm:fourier_strongly_explicit} in multiple ways. First, the construction runs in sublinear time in $p$ but it is not strongly explicit. Second, it gives different trade-offs between the sparsity parameter and the number of rows. Last but not least, the construction depends on the factorization of $p-1$. 

\begin{theorem}[Incoherent matrices from DFT via multiplicative subgroups, Section~\ref{sec:weil}] \label{thm:fourier_incoherent_subgroup}
Let $p$ be a prime number. For every divisor $d$ of $p-1$ such that $d> \sqrt{p}$ we can find in time $O(d \log p)$ a matrix $M \in \C^{ d \times p}$ with rows being the rows of the DFT matrix such that $\frac{1}{d}M$ is $(\sqrt{p} /d)$-incoherent.
\end{theorem}

This result could give (depending on the factorization of $p-1$) a better polynomial dependence of $m$ on $k$ in the high-sparsity regime. If $p-1$ has a large divisor about $p^{1-\gamma}$, this would yield a matrix with sparsity parameter $k \approx p^{\gamma}$ and $m \approx k^{1/\gamma-1}$ rows. For example, when $\gamma=1/4$, we obtain $k\approx p^{1/4}$ and $m \approx k^3$, which cannot be obtained from Theorem~\ref{thm:fourier_strongly_explicit}. In general, Theorem~\ref{thm:fourier_incoherent_subgroup} will yield useful matrices as long as $p-1$ has divisors in the range $[\sqrt{p}, p-1]$, ideally as many as possible. An extreme case is Fermat primes, which have $(\log p)/2$ divisors in the aforesaid interval.

The reader might ask the question if the polynomial dependence of $k$ on $p$ is necessary; ideally one would like a logarithmic dependence, since the polynomial dependence is interesting only in the high-sparsity regime. Regarding strongly explicit constructions, we provide some evidence why this might be a very hard problem in the remark below. 

\begin{remark}\label{rem:justification}
The inferiority of our bounds in the low-sparsity regime is justifiable to some extent: it is because of a common obstacle that has persisted more than a century in the theory of exponential sums, due to the lack of techniques to account for sparse character sums (either additive or multiplicative). In general, the fewer summands the sum has, the harder it is to prove a tight cancellation bound. Thus, owing to the use of heavy machinery from analytic number theory and more specifically the theory of exponential sums over finite fields, our bounds for strongly explicit constructions are quite suboptimal.
\end{remark}

\subsection{Comparing $\ell_2/\ell_1$ with $\ell_{\infty}/\ell_1$} \label{sec:comparison}

In this subsection we elaborate why $\ell_\infty/\ell_1$ is much stronger than $\ell_2/\ell_1$, and not just a guarantee that implies $\ell_2/\ell_1$. Let $\gamma < 1$ be a constant and consider the following scenario. There are three sets $A,B,C$ of size $\gamma k, (1-\gamma)k$, $n - k$ respectively, and for every $i\in A$ we have $|\wh{x}_i| = \frac{2}{k} \|\wh{x}_{C} \|_1 = \frac{2}{k} \|\wh{x}_{-k}\|_1$, while every coordinate in $B$ and $C$ has the equal magnitude. It follows immediately that
\[
\|\wh{x}_C\|_1 = \frac{n - k  }{n-\gamma k } \|\wh{x}_{B \cup C}\|_1.
\]

Now assume that $k\leq \gamma n$, then $(n-\gamma k)/(n-k)\leq 1+\gamma$. We claim that the zero vector is a valid solution for the $\ell_2/\ell_1$ guarantee, since 
\begin{align*}
\|\vec 0 - \wh{x}\|_2^2  &= \|\wh{x}_A\|_2^2 + \|\wh{x}_{B\cup C}\|_2^2\\
&\leq \gamma k\cdot \frac{4}{k^2} \|\wh{x}_{-k}\|_1^2 + \frac{1}{(n-\gamma k)} \|\wh{x}_{B\cup C}\|_1^2\\
&\leq \frac{4\gamma}{k} \|\wh{x}_{-k}\|_1^2 + \frac{n-\gamma k}{(n-k)^2}\|\wh{x}_C\|_1^2 \\
&\leq \left(\frac{4\gamma}{k} +  \frac{1+\gamma}{n-k}\right)\|\wh{x}_{-k}\|_1^2 \\
&\leq \frac{5\gamma}{k} \|\wh{x}_{-k}\|_1^2,
\end{align*}
where the last inequality follows provided it further holds that $k\leq \gamma n /(2\gamma + 1)$. Hence when $\gamma\leq 1/5$, we see that the zero vector satisfies the $\ell_2/\ell_1$ guarantee.

Since $\vec 0$ is a possible output, we may not recover any of the coordinates in $S$, which is the set of ``interesting'' coordinates. 
On the other hand, the $\ell_{\infty}/\ell_1$ guarantee does allow the recovery of \textbf{every} coordinate in $S$. 
This is a difference of recovering all $\gamma k$ versus $0$ coordinates. We conclude from the discussion above that in the case of too much noise, the $\ell_2/\ell_1$ guarantee becomes much weaker than the $\ell_\infty/\ell_1$, possibly giving meaningless results in some cases.

%% file: Overview.tex
\section{Overview}

\paragraph{Sparse Fourier Transform Algorithms (Subsection~\ref{subsec:sft}).}
We first show how to achieve the for-all schemes, i.e., schemes that allow universal reconstruction of all vectors, and then derandomize them. Similarly to the previous works \cite{hikp12b,ik14,k17}, our algorithm hashes, with the filter in \cite{k17}, the spectrum of $x$ to $O(k)$ buckets using pseudorandom permutations, and repeat $O(k \log n)$ times with fresh randomness. The main part of the analysis is to show that for any vector $\hat{x} \in \mathbb{C}^n$ and any set $S \subseteq [n]$ with $|S| \leq k$, each $i\in S$, in a constant fraction of the repetitions, receives ``low noise'' from all other elements, under the pseudorandom permutations. This will boil down to a set of $\Theta(n^2)$ inequalities involving the filter and the pseudorandom permutations. We prove these inequalities with full randomness (Lemma~\ref{lem:initial_constraint}), and then derandomize the pseudorandom permutations using the method of conditional expectations (Lemma~\ref{lem:derandomization_step}). This will give us Theorem~\ref{thm:intro_main_1}. To do so, we choose the pseudorandom permutations one at a time, repetition by repetition, and keep an (intricate) pessimistic estimator (Lemma~\ref{lem:pessimistic}), which we update accordingly. Our argument extends the arguments in~\cite{nnw14} and~\cite{porat2008explicit}, and could be of independent interest. To compare with~\cite{nnw14} we have the following observation. The construction in \cite{nnw14} consists of $O(k\log n)$ matrices, joined vertically, each having $O(k)$ rows and exactly one $1$ per column. This ensures a small incoherence of the concatenated matrix and gives the $\ell_\infty/\ell_1$ guarantee. In the Fourier case, the convolution with the filter functions behaves analogously: instead of having exactly one non-zero element, each  column in the $\ell$-th matrix has a contiguous segment of $1$s of size $\approx n/k$ (where the center of that segment  depends on the choice of the $\ell$-th pseudorandom permutation) and polynomially decaying entries away from this segment. Moreover, the positions of the segments across the columns are not fully independent and are defined via the pseudorandom permutations in Definition~\ref{def:pseudopermute}. We show that even in this more restricted setting, derandomization is possible in polynomial time. Several details are omitted in the preceding high-level discussion and we suggest the reader look at the corresponding sections for the complete argument.

The sublinear-time algorithm (Theorem~\ref{thm:intro_main_2}) is obtained by bootstrapping the derandomized scheme above with an identification procedure in each bucket, as most previous algorithms have done (e.g.~\cite{hikp12a}). The major difference is that our identification procedure needs to be deterministic. We show an explicit set of samples that allow the implementation of the desired routine. To illustrate our idea, let us focus on the following $1$-sparse case: $\hat{x}\in \mathbb{C}^n$ and  $|\wh{x}_{i^*}| \geq 3 \|\wh{x}_{[n]\setminus i^\ast}\|_1$ for some $i^*$, which we want to locate. Let 
\[	
	\theta_j = \left( \frac{2\pi}{n} j \right) \mathrm{mod} ~ 2\pi,
\] 
and consider the $\log n  $ samples $x_0, x_1, x_2, x_4,\ldots,x_{2^{r-1}},\dots$.

Observe that (ignoring $1/\sqrt{n}$ factors)

 	\[	x_{\beta} = \wh{x}_{i^\ast} e^{\sqrt{-1} \beta \theta_{i^\ast}} + \sum_{ j \neq i^\ast} \wh{x}_j e^{\sqrt{-1} \beta \theta_j},	\]
we can find $\beta \theta_{i^*} + \arg \wh{x}_{i^\ast}$ up to $\pi/8$, just by estimating the phase of $x_{\beta}$ and Proposition~\ref{fact:phase_est}. Thus we can estimate $\beta\theta_{i^\ast}$ up to $\pi/4$ from the phase of $x_{\beta}/x_0$.  If $i^\ast \neq j$, then there exists a $\beta \in \{1,2,2^2, \dots, 2^{r-1}, \ldots\}$ such that $|\beta \theta_{i^*} - \beta \theta_j|_{\circ} > \pi/2$, and so $\beta \theta_j$ will be more than $\pi/4$ away from the phase of the measurement. Thus, by iterating over all $j \in [n]$, we keep the index $j$ for which $\beta \theta_j$ is within $\pi/4$ from $\arg(x_{\beta}/x_0)$, for every $\beta$ that is a power of $2$ in $\mathbb{Z}_n$.

Unfortunately, although this is a deterministic collection of $O(\log n)$ samples, the above argument gives only $O(n \log n)$ time. For sublinear-time decoding we use $x_1/x_0$ to find a sector $S_0$ of the unit circle of length $\pi/4$ that contains $\theta_{i^*}$. Then, from $x_2/x_0$ we find two sectors of length $\pi/8$ each, the union of which contains $\theta_{i^*}$. Because these sectors are antipodal on the unit circle, the sector $S_0$ intersects exactly one of those, let the intersection be $S_1$. The intersection is a sector of length at most $\pi/8$. Proceeding iteratively, we halve the size of the sector at each step, till we find $\theta_{i^*}$, and infer $i^*$. Plugging this idea in the whole $k$-sparse recovery scheme yields the desired result. Our argument crucially depends on the fact that in the $\ell_1$ norm the phase of $\theta_{i^\ast}$ will always dominate the phase of all samples we take. 

\paragraph{Incoherent Matrices from the Fourier ensemble (Subsection~\ref{subsec:incoherent}).}

Our first result for incoherent matrices (Theorem~\ref{thm:incoherent1}) is more general and works for any matrix that has orthonormal columns with entries bounded by $O(1/\sqrt{n})$. We subsample the matrix, invoke a Chernoff bound and Bernstein's inequality to show the small incoherence of the subsampled matrix. We follow a derandomization procedure which essentially mimics the proof of Bernstein's inequality, keeping a pessimistic estimator which corresponds to the sum of the generating functions of the probabilities of all events we want to hold, evaluated at specific points. We obtain an explicit construction, i.e. a derandomization in $\mathrm{poly}(n)$ time. This argument could be of independent interest for its generality. As there are many technical obstacles to overcome, we suggest the reader take a careful look at the proof to gain a clearer picture of the argument.
 
Our next results (Theorem~\ref{thm:fourier_strongly_explicit} and Theorem~\ref{thm:fourier_incoherent_subgroup}) construct \emph{strongly explicit} incoherent matrices by making use of technology from the fruitful theory of exponential sums in analytic number theory and additive combinatorics. Roughly speaking, to bound a complex exponential sum over a set $S$, one would expect that specific choices of the set $S$ lead to non-trivial bounds, i.e. $o(|S|)$, since cancellation takes place in the summation. Ideally, one would desire that the exponentials behave like a random walk and give the optimal cancellation of $O(\sqrt{|S|})$. This intuition is clearly not true, but the results by Weyl and others show that certain sets $S$ can exhibit a nicer behaviour. We exploit their results to build incoherent matrices by taking the rows of the DFT matrix indexed by the ``nice'' sets. This connection also yields an immediate improvement on the lower bound of an exponential sum obtained by Winterhof~\cite{winterhof2001}.

%% file: preliminaries.tex
\section{Technical Toolkit}

\subsection{Hash Functions}
\begin{definition}[Frequency domain hashings $\pi,h,o$] \label{def:pho}
Given $\sigma,b \in [n]$, we define a function $\pi_{\sigma,b} : [n] \rightarrow [n]$ to be
$\pi_{\sigma,b} (f) = \sigma ( f - b ) \pmod n$ for all $f \in [n]$. Define a hash function $h_{\sigma,b} : [n] \rightarrow [B]$ as
$h_{\sigma,b}(f) = \operatorname{round} ( (B/n) \pi_{\sigma,b}(f)  )$
and the off-set functions $o_{f,\sigma,b} : [n] \rightarrow [n/B]$ as $o_{f,\sigma,b}(f') = \pi_{\sigma,b}(f') - (n/B) h_{\sigma,b}(f)$. 
When it is clear from context, we will omit the subscripts $\sigma,b$ from the above functions.
\end{definition}

In what follows, we might use the notation $H = (\sigma,a,b)$ to denote a tuple of values along with the associated hash function from Definition \ref{def:pho}. Below we define a pseudorandom permutation in the frequency domain.

\begin{definition}[$P_{\sigma,a,b}$] \label{def:pseudopermute}
Suppose that $\sigma^{-1}\mod n$ exists. For $a,b \in [n]$, we define the pseudorandom permutation $P_{\sigma,a,b}$ by
$(P_{\sigma,a,b} x)_t = x_{\sigma (t-a)} \omega^{t \sigma b}$.
\end{definition}

\begin{proposition}[{\cite[Claim 2.2]{hikp12a}}]
$(\wh{ P_{\sigma,a,b} x })_{ \pi_{\sigma,b} ( f ) } = \wh{x}_f \omega^{a \sigma f}$.
\end{proposition}

\begin{definition} [Sequence of Hashings]
A sequence of $d$ hashings is specified by $d$ tuples $\{ (\sigma_r,a_r,b_r) \}_{r \in [d]}$. For a fixed $r \in [d]$, we will also set $\pi_r , h_r, o_r$ to be the functions defined in Definition \ref{def:pho}, and $P_r$ to be the pseudorandom permutation defined in Definition \ref{def:pseudopermute}, by setting $a = a_r, b = b_r, \sigma= \sigma_r$.
\end{definition}

\subsection{Filter Functions}
\begin{definition}[Flat filter with $B$ buckets and sharpness $F$~\cite{k17}]\label{def:filter_G}
A sequence $\wh{G} \in \R^n$ symmetric about zero with Fourier transform $G \in \R^n$ is called a flat filter with $B$ buckets and sharpness $F$ if \\
(1) $\wh{G}_f \in [0,1]$ for all $f\in [n]$; \\
(2) $\wh{G}_f \geq 1 - (1/4)^{F-1}$ for all $f \in [n]$ such that $|f| \leq \frac{n}{2B}$; \\
(3) $\wh{G}_f \leq (1/4)^{F-1} ( \frac{n}{B |f|} )^{F-1}$ for all $f \in [n]$ such that $|f| \geq \frac{n}{B}$.
\end{definition}

\begin{lemma}[Compactly supported flat filter with $B$ buckets and sharpness $F$ \cite{k17}]
Fix the integers $(n,B,F)$ with $n$ a power of two, integers $B < n$, and $F \geq 2$ an even integer. There exists an $(n,B,F)$-flat filter $\wh{G} \in \R^{n}$, whose inverse Fourier transform $G$ is supported on a length-$O(FB)$ window centered at zero in time domain.
\end{lemma}

\begin{lemma}[{\cite[Lemma 3.6]{hikp12b}}, {\cite[Lemma 2.4]{hikp12a}}, {\cite[Lemma 3.2]{ik14}}]
Let $f,f' \in [n]$. Let $\sigma$ be uniformly random odd number between $1$ and $n-1$. Then for all $d \geq 0$ we have $\Pr[ |\sigma(f-f')|_{\circ} \leq d ] \leq 4 d /n$.
\end{lemma}

\subsection{Formulas for Estimation}

\begin{definition}[Measurement]
For a signal $\wh{x} \in \C^n$, a hashing $H = (\sigma,a,b)$, integers $B$ and $F$, a measurement vector $m_{H}\in \C^B $ is the $B$-dimensional complex-valued vector such that
 \begin{align*}
(m_H)_s = \sum_{f \in [n]} \wh{G}_{\pi(f) - (n/B) \cdot s} \omega^{a \sigma f} \cdot \wh{x}_f  \in \C
\end{align*}
for $ s \in [B]$. Here $\wh{G}$ is a filter with $B$ buckets and sharpness $F$ constructed in Definition~\ref{def:filter_G}. \end{definition}

The following lemma provides a \textsc{HashToBins} procedure, which computes the bucket values of the residual $\hat x - \hat z$, where $\hat z$ is also provided as input.

\begin{lemma}[\textsc{HashToBins}~{\cite[Lemma 2.8]{k17}}]\label{lem:hashtobins}
Let $H = (\sigma, a, b)$ and parameters $B,F$ such that $B$ is a power of $2$, and $F$ is an even integer. There exists a deterministic procedure \textsc{HashToBins}$(x, \wh{z}, H)$ which computes $u \in \C^B$ such that for any $f \in [n]$,
\begin{align*}
u_{h(f)} = \Delta_{h(f)} + \sum_{f' \in [n]} \wh{G}_{o_f(f')} ( \wh{x} - \wh{z} )_{f'} \omega^{a \sigma f'} ,
\end{align*}
where $\wh{G}$ is the filter defined in Definition~\ref{def:filter_G}, and  $\Delta_{h(f)}$ is a negligible error term satisfying $|\Delta_{h(f)}|\leq \|z\|_2\cdot n^{-c}$ for $c>0$ an arbitrarily large absolute constant. It takes $O(BF)$ samples, and $O(F \cdot B \log B + \| \wh{z} \|_0 \cdot \log n)$ time.
\end{lemma}

We shall ignore the $\Delta_{h(f)}$ term in the proof of correctness of our algoriths, since it will be negligible and won't affect the analysis.
For a hashing $H = (\sigma, a, b)$, values $B,F$, and the associated measurement $m_H$, one has 

\begin{equation}\label{eqn:bucket}
\wh{G}_{o_f(f)}^{-1} (m_H)_{h(f)} \omega^{ - a \sigma f } = \wh{x}_f + \underbrace{ \wh{G}_{ o_f(f) }^{-1} \sum_{f' \in [n] \backslash \{f\} } \wh{G}_{o_f(f')} \wh{x}_{f} \omega^{a \sigma (f'-f)} }_{ \text{noise~term} }.
\end{equation}

The following is a basic fact of complex numbers, which will be crucially used in our sublinear-time algorithm, for estimating the phase of a heavy coordinate.
\begin{proposition}\label{fact:phase_est}
Let $x,y\in \C$ with $|y|\leq |x|/3$, then $|\arg(x+y)-\arg x| \leq \pi/8$.
\end{proposition}
\begin{proof}
The worst case occurs when $y$ is orthogonal to $x$, and thus $|\arg(x+y)-\arg x| \leq \arctan(1/3) < \pi/8$.
\end{proof}

%% file: linear_algo.tex
\section{Linear-Time Algorithm}\label{sec:linear_time}

Our first step is to obtain a condition that allows us to approximate every coordinate of $\wh{x} \in \mathbb{C}^n$. This condition corresponds to a set of $n(n-1)$ inequalities. In this section we shall consider a sequence of hashings $\{H_r\}_{r \in [d]} = \{ (\sigma_r,a_r,b_r)\}_{r \in [d]}$ and for notational simplicity we shall abbreviate $o_{f,\sigma_r,b_r}(f')$ as $o_{f,r}(f')$.

We first present a lemma, which states that each $\wh{x}_f$ can be finely estimated in most hashing repetitions.
\begin{lemma}
\label{lem:G_constraint}
Fix $B$ and $F$. Let a sequence of hashings $\{H_r\}_{r \in [d]} = \{ (\sigma_r,a_r,b_r)\}_{r \in [d]}$ and $x \in \mathbb{C}^n$. If for all $f,f' \in [n]$ with $f\neq f'$ it holds that 
\begin{equation}\label{eqn:G_constraint_1}
	\sum_{r \in [d]} \wh{G}^{-1}_{o_{f,r}(f)} \wh{G}_{o_{f,r}(f')} \leq \frac{2d}{B},	
\end{equation}
then for every vector $x \in \mathbb{C}^n$ and every $f\in [n]$, for at least $8d/10$ indices $r \in [d]$ we have that
\begin{equation}\label{eqn:estimate_error}	
\left|\wh{x}_{f} -\wh{G}_{o_{f,r}(f)}^{-1}  (m_{H_r})_{h_r(f)} \right| \leq \frac{10}{B} \|\wh{x}_{[n] \setminus \{f\}}\|_1.
\end{equation}

\end{lemma}

\begin{table}
\begin{center}
    \begin{tabular}{| l | l | l |}
    \hline
	Notation & Semantics     \\ \hline \hline
	$C$  & Absolute Constant \\ \hline
	$B$  & Number of ``Buckets'', power of $2$\\ \hline
	$d$ & Number of ``repetitions''\\ \hline
	$\beta$ & equals $CB/d$ \\\hline
       	$\gamma$ & Rate of SNR decrease \\ \hline
	$\mu$	& Given Approximation to SNR \\\hline
	$\nu^{(t)}$ & Approximation of SNR at the $t$-th step \\\hline
	$r^{(t)}$ & Residual at the $t$-th step \\\hline
\end{tabular}
\end{center}
\caption{Notation and semantics for variables in this subsection.}

\end{table}

\begin{proof}
We have that 
\begin{align*}
\sum_{r \in [d]} \left|\wh{x}_f- \wh{G}_{o_{f,r}(f)}^{-1} (m_{H_r})_{h_r(f)} \right| 
&=  \sum_{r \in [d]} \left |\wh{G}^{-1}_{o_{f,r}(f)} \sum_{f' \in [n] \setminus \{ f \} } \wh{G}_{o_{f,r}(f')} \wh{x}_{f'} \omega^{a_r \sigma_r(f'-f)} \right| \qquad \text{(by \eqref{eqn:bucket})}\\
&\leq \sum_{r \in [d]}  \wh{G}^{-1}_{o_{f,r}(f)}  \sum_{f' \in [n] \setminus \{f\}}  \wh{G}_{o_{f,r}(f')} |\wh{x}_{f'} | \\
&= \sum_{f' \in [n] \setminus \{f \} } |\wh{x}_{f'}| \sum_{r \in [d]} \wh{G}^{-1}_{o_{f,r}(f)} \wh{G}_{o_{f,r}(f')} \\
&\leq \sum_{f' \in [n] \setminus \{f\}} |\wh{x}_{f'}| \frac{2d}{B}.
\end{align*}
Hence there can be at most $2d/10$ indices $r \in [d]$ for which the estimate $|\wh{x}_f - \wh{G}_{o_{f,r}}(f)  \cdot m_r(h_r(f))  |$ is more than $(10/B)\|\wh{x}_{[n]\setminus\{f\}}\|_1$, otherwise the leftmost-hand side would be at least
 $(2d/10 +1) \cdot (10/B)\|\wh{x}_{[n]\setminus\{f\}}\|_1 > 2(d/B)\|\wh{x}_{[n]\setminus\{f\}}\|_1$.
\end{proof}

The lemma above implies that for every $f \in [n]$ we can find an estimate of $\wh{x}_f$ up to $\frac{10}{B} \|\wh{x}_{[n] \setminus \{f\}}\|_1$ in time $O(d)$, by taking the median of all values $m_r(h_r(f))$ for $r \in [d]$. The existence of pseudorandom permurations such that the conditions of Lemma~\ref{lem:G_constraint} hold, namely inequalities~\ref{eqn:G_constraint_1}, is proved in Lemma~\ref{lem:initial_constraint}, see next subsections for notation and definitions.

\subsection{Proof of correctness assuming Inequalities~\eqref{eqn:G_constraint_1} hold} 
We prove the first part of Theorem~\ref{thm:intro_main_1} (existence of $S$) assuming that the inequalities~\ref{eqn:G_constraint_1} hold, and thus the conditions of Lemma~\ref{lem:G_constraint} hold.

For notational simplicity, let $\epsilon = (1/4)^{F-1}$ so the filter $\wh{G}$ satisfies that $\wh{G}_{f'}\geq 1-\epsilon$ for all $f'\in [-\frac{n}{2B},\frac{n}{2B}]$ and $\wh{G}_{f'}\leq \epsilon$ for all $f'\in [n]\setminus (-\frac{n}{B},\frac{n}{B})$. In the rest of the section, we choose $B = 10(1-\eps)^{-1}\beta k$ rounded to the closest power of $2$ from above; $\beta$ is some constant to be determined. 

As in previous Fourier sparse recovery papers~\cite{hikp12a,ik14,k16,k17}, we assume that we have the knowledge of $\mu = \|\wh{x}_{-k}\|_1/k$ (or a constant factor upper bound) and that the signal-to-noise ratio $R^\ast = \|\wh{x}\|_1 / \mu \leq n^\alpha$. Our estimation algorithm is similar to that in~\cite{ik14}. The main algorithm is Algorithm~\ref{alg:recovery_overall}. It recovers the heavy coordinates of $\hat x$ in increasing magnitude by repeatedly calling the subroutine Algorithm~\ref{alg:linear_recovery}, which recovers the heavy coordinates of the residual spectrum above certain threshold.

The following lemmata are analogous to Lemmata 6.1 and 6.2 in \cite{ik14}, and their proofs are postponed to Section~\ref{sec:omitted_proofs}. The first lemma states that Algorithm~\ref{alg:linear_recovery} will recover all the coordinates in the residual spectrum that are at least $\nu$ and it will not mistake a small coordinate for a large one.

\begin{algorithm}[tb]
\caption{Overall algorithm}\label{alg:recovery_overall}
\begin{algorithmic}
\State $T \gets \log_\gamma R^\ast$	\Comment{$\gamma$ is an absolute constant}
\State $\wh{z}^{(0)} \gets 0$
\State $\nu^{(0)} \gets C \mu$
\For{$t=0$ to $T-1$}
	\State $\wh{z}^{(t+1)} \gets \wh{z}^{(t)} + \Call{SubRecovery}{x,\wh{z}^{(t)},\nu^{(t)}}$
	\State $\nu^{(t+1)} \gets \gamma \nu^{(t)}$
\EndFor
\State \Return $\wh{z}$
\end{algorithmic}
\end{algorithm}

\begin{algorithm}[tb]
\caption{Linear-time Sparse Recovery for $\wh{x}-\wh{z}$}\label{alg:linear_recovery}
\begin{algorithmic}
\Function{SubRecovery}{$x,\wh{z},\nu$}
\State $S \gets \emptyset$
\For{$r=1$ to $d$}
	\State $u_r \gets \textsc{HashToBins}(x,\wh{z},(\sigma_r,0,b_r))$
\EndFor
\For{$f \in [n]$}
	\State $\hat{x}_f' = \median_{r \in [d]} \wh{G}_{o_{f,r}(f)}^{-1}(u_r)_{h_r(f)}$ \Comment{ $o_{f,r} = o_{f,\sigma_r,b_r}$}
	\If {$|\hat{x}_f'| > \nu/2$}
		\State $S \gets S \cup \{ f \}$
	\EndIf
\EndFor
\State \Return $\hat{x}_S'$
\EndFunction
\end{algorithmic}
\end{algorithm}

\begin{lemma}[guarantee of $\textsc{SubRecovery}$, Section~\ref{sec:omitted_proofs}]\label{lem:recovery_of_HH}
Consider the call $\textsc{SubRecovery}\{x,\wh{z},\nu\}$ (Algorithm~\ref{alg:linear_recovery}). Let $w = \hat x -\hat z$. When $\nu\geq \frac{16}{\beta k}\|\wh{w}\|_1$, the output $\wh{w}'$ of Algorithm~\ref{alg:linear_recovery} satisfies
\begin{enumerate}[label=(\roman*),noitemsep]
	\item $|\wh{w}_f| \geq (7/16)\nu$ for all $f \in \supp(\wh{w}')$.
	\item $|\wh{w}_f-\wh{w}'_f| \leq |\wh{w}_f|/7$ for all $i \in \supp(\wh{w}')$;
	\item $\supp(\wh{w}')$ contains all $f$ such that $|\wh{w}_f|\geq \nu$;
\end{enumerate}
\end{lemma}

Next we turn to the analysis of Algorithm~\ref{alg:recovery_overall}.
Let $H = H(\wh{x},k)$ and $I = \{f: |\wh{x}_f|\geq \frac{1}{\rho k}\|\wh{x}_{-k}\|_1 \}$ for  some constant $\rho$ to be determined. By the SNR assumption of $\wh{x}$, we have that $\|\wh{x}_{H}\|_1\leq k\|\wh{x}\|_\infty \leq R^\ast\|\wh{x}_{-k}\|_1$ and thus $\|\wh{x}\|_1\leq (R^\ast+1)\|\wh{x}_{-k}\|_1$. In Algorithm~\ref{alg:recovery_overall}, the threshold in the $t$-th step is
\[
\nu^{(t)} = C\mu\gamma^{T-t},
\]
where $C \geq 1, \gamma > 1$ are constants to be determined. Let $r^{(t)}$ be the residual vector at the beginning of the $t$-th step in the iteration. 
We can show that the coordinates we shall ever identify are all heavy (contained in $I$) and we always have good estimates of them.

\begin{lemma}[$\ell_\infty$ norm reduction, Section~\ref{sec:omitted_proofs}]\label{lem:main_loop_invariant}
There exist $C,\beta,\rho,\gamma$ such that it holds for all $0\leq t \leq T$ that
\begin{enumerate}[label=(\alph*),noitemsep]
	\item $\wh{x}_f = r^{(t)}_f$ for all $f\notin I$;
	\item $|r^{(t)}_f| \leq |\wh{x}_f|$ for all $f$.
	\item $\|r^{(t)}_I\|_\infty \leq \nu^{(t)}$;
\end{enumerate}
\end{lemma}

Now we are ready to show the first part of Theorem~\ref{thm:intro_main_1}, which is one of our main results. We shall choose $d = O(k\log n)$ such that the conditions in \eqref{lem:G_constraint} holds. The hashings $\{H_r\}_{r\in [d]}$ can be chosen deterministically, which we shall prove in the rest of the section after this proof; this will complete the full proof.

\begin{proof}[Proof of Theorem~\ref{thm:intro_main_1}]

The recovery guarantee follows immediately from Lemma~\ref{lem:main_loop_invariant}, as
\begin{equation}\label{eqn:linear_final}
\|r^{(T)}\|_\infty \leq \max\{ \|r^{(T)}_I\|_\infty, \|r^{(T)}_{I^c}\|_\infty\} 
\leq \max\{ \nu^{(T)}, \|\wh{x}_{I^c}\|_\infty \} 
\leq \max\{ 2\mu, (1/\rho)\mu \} 
= 2\mu.
\end{equation}
This implies that $\|\hat x-\hat x'\|_\infty \leq (2/k)\|x_{-k}\|_1$. To obtain the $\ell_\infty/\ell_1$ error guarantee, that is, to achieve a right-hand side of $(1/k)\|x_{-k}\|_1$, we can just replace $k$ with $2k$ throughout our construction and analysis.

\paragraph{Number of Measurements.} Computing the measurements in \textsc{SubRecovery} requires $O(k)$ measurements (Lemma~\ref{lem:hashtobins}). These measurements are reused throughout the iteration in the overall algorithm, hence there are $O(k d) = O(k \cdot k\log n) = O(k^2\log n)$ measurements in total.

\paragraph{Running Time.} Each call to \textsc{SubRecovery} runs in time $O(d(B\log B + \|\hat z\|_0\log n) + nd) = O(k^2\log k \log n + k\|\hat z\|_0\log^2 n + n k\log n)$. By Lemma~\ref{lem:main_loop_invariant}(a), we know that $\|\hat z\|_0 \leq |I| = O(k)$. The overall runtime is therefore $O(k^2\log k\log n + n k\log^2 n + k^2\log^2 n) = O(k^2\log^2 n + n k\log^2 n) = O(n k\log^2 n)$.

\end{proof}

\subsection{Choosing the hash functions}

In this and the next subsection, we shall find $\{(\sigma_r,a_r,b_r)\}_{r\in [d]}$ such that \eqref{eqn:G_constraint_1} holds for all pairs $f\neq f'$. It will be crucial for the next section that we can choose $a_r$ freely; that means the inequalities depend solely on $\sigma_r,b_r$. Note that $o_{f,r}(f)\in [-\frac{n}{2B},\frac{n}{2B}]$ and thus $\hat G_{o_{f,r}(f)}\in [1-\epsilon,1]$, it suffices to find $\{(\sigma_r,b_r)\}_{r\in [d]}$ such that it holds for all $f\neq f'$ that
\[
	\sum_{r \in [d]} \wh{G}_{o_{f,r}(f')} \leq \frac{2}{1+\epsilon}\cdot \frac{d}{B}.
\]

We shall show how to do so in polynomial time in $n$.

\begin{definition} [Bad Events]
Let $C=2/(1+\eps)$ and $\beta = Cd/B$. Let  $A_{f,f'}$ denote the event $\sum_{r=1}^d \wh{G}_{o_{f,r}(f')}\geq \beta$.
\end{definition}

\paragraph{Pessimistic Estimator}
The derandomization proceeds as follows: find a pessimistic estimator $h_r(f,f';\sigma_1,b_1,\dots,\sigma_r,b_r)$ for each $r$ with the first $r$ hash functions fixed by $(\sigma_1,b_1),\dots,(\sigma_r,b_r)$ such that the following holds:
\begin{gather}
	\Pr\left(A_{f,f'} | \sigma_1,b_1,\dots,\sigma_r,b_r\right) \leq h_r(f,f';\sigma_1,b_1,\dots,\sigma_r,b_r) \label{eqn:pessimistic}\\
	\sum_{f\neq f'} h_0(f,f') < 1 \label{eqn:pessimistic_initial}\\
	h_r(f,f';\sigma_1,b_1,\dots,\sigma_r,b_r) \geq \E_{\sigma_{r+1},b_{r+1}}h_{r+1}(f,f';\sigma_1,b_1,\dots,\sigma_r,b_r,\sigma_{r+1},b_{r+1}) \label{eqn:pessimistic_step}
\end{gather}

Note that inequality~\ref{eqn:pessimistic_initial} implies that there exist choices of the pseudorandom permutations such that the conditions of Lemma~\ref{lem:G_constraint} hold.
The algorithm will start with $r=0$. At the $r$-th step, it chooses $\sigma_{r+1},b_{r+1}$ to minimize 
\[
\sum_{f\neq f'} h_{r+1}(f,f';\sigma_1,b_1,\dots,\sigma_r,b_r,\sigma_{r+1},b_{r+1}).
\] 
By \eqref{eqn:pessimistic_step}, this sum keeps decreasing as $r$ increases. At the end of step $d-1$, all hash functions are fixed, and by \eqref{eqn:pessimistic} and \eqref{eqn:pessimistic_initial}, we have
$\sum_{f\neq f'} \Pr(A_{f,f'}| \sigma_1,b_1,\dots,\sigma_d,b_d) < 1$.
Since $A_{f,f'}$ is a deterministic event conditioned on all $d$ hash functions, the conditional probability is either $0$ or $1$. The inequality above implies that all conditional probabilities are $0$, i.e., none of the bad events $A_{f,f'}$ happens, as desired. 

We first define our pessimistic estimator. In what follows, we shall be dealing with numbers that might have up to $O(n)$ digits. Manipulating numbers of that length can be done in polynomial time. We will not bother with determining the exact exponent in the polynomial or optimizing it, which we leave to future work.

\begin{definition}[Pessimistic Estimator]
Let $\lambda>0$ to be determined. Define 
\[
h_r(f,f';\sigma_1,b_1,\dots,\sigma_r,b_r) = e^{-\lambda\beta} \exp\left(\lambda\sum_{\ell=1}^r \widehat{G}_{o_{f,\ell}(f')} \right) (M(\lambda))^{d-r},
\]
where
\[
M(\lambda) = e^{\lambda\eps}\left[\left(\frac{2}{B}+\frac{1}{n}\right)(e^{\lambda(1-\epsilon)}-1) + 1\right].
\]
\end{definition}

This function can be evaluated in $\tilde O(r) \cdot \mathrm{poly}(n)$ time for each pair $f\neq f'$ and thus the algorithm runs in polynomial time in $n$.

To complete the proof, we shall verify \eqref{eqn:pessimistic}--\eqref{eqn:pessimistic_step} in Subsection~\ref{subsec:finishing}. 

\subsection{Distribution of Offset Function}
This subsection prepares auxiliary lemmata which will be used to verify the derandomization inequalities. In this subsection we focus on the distribution of the offset $o_{f,\sigma,b}(f')$ for $f'\neq f$ and appropriately random $\sigma$ and $b$.

\begin{lemma}\label{lem:dist_o_f(f')}
Suppose that $n,B$ are powers of $2$, $\sigma$ is uniformly random on the odd integers in $[n]$ and $b$ is uniformly random in $[n]$. For any fixed pair $f\neq f'$ it holds that
\begin{enumerate}[label=(\roman*)]
	\item When $(n/B)\nmid (f-f')$, $o_{f,\sigma,b}(f')$ is uniformly distributed on $[n]$;
	\item When $(f-f')/(n/B)$ is even, $\Pr\{o_{f,\sigma,b}(f') = \ell\} = 0$ for all $\ell \in [-\frac{n}{B},\frac{n}{B}]$.
	\item When $(f-f')/(n/B)$ is odd, $\Pr\{o_{f,\sigma,b}(f') = \ell\} = 0$ for $\ell \in [-\frac{n}{2B},\frac{n}{2B})$ and $\Pr\{o_{f,\sigma,b}(f') = \ell\} = \frac{2}{n}$ for $\ell \in [-\frac{n}{B},-\frac{n}{2B})\cup [\frac{n}{2B} , \frac{n}{B}]$.
\end{enumerate}
\end{lemma}
\begin{proof}
First observe that
\[
o_{f,\sigma,b}(f') \equiv \sigma(f'-f) + \sigma(f-b) - \frac{n}{B}\operatorname{round}\left(\frac{B}{n}\sigma(f-b)\right) \pmod n.
\]
For a fixed $\sigma$, let
\[
Z_\sigma = \sigma(f-b) - \frac{n}{B}\operatorname{round}\left(\frac{B}{n}\sigma(f-b)\right).
\]
Note that $\sigma(f-b) \bmod n$ as a function of $b$ is uniform on $[n]$. Note also that

\[Z_\sigma = \frac{n}{B}\left(\frac{B}{n} \sigma(f-b) - \operatorname{round}\left(\frac{B}{n}\sigma(f-b)\right) \right),	\]

which gives that $Z_\sigma$ is uniform on its support, which is $[-\frac{n}{2B},\frac{n}{2B})$. 

Suppose that $f'-f \equiv 2^s K\pmod{n}$, where $K\geq 1$ is an odd integer. It is clear that $\sigma(f'-f)$ is uniform on its support $T=\{2^s \ell \bmod n: \ell \text{ is odd}\}$, which consists of equidistant points. Since $Z_\sigma$ is always uniform (regardless of $\sigma$), and the distribution of $o_{f}(f') = \sigma(f-f') + Z_\sigma$ is the convolution of two distributions.

Suppose that now that $n = 2^r$ and $B=2^b$. 

When $(n/B)\nmid (f'-f)$, it holds that $r-b\geq s+1$, and thus $n/B$ is an integer multiple of the distance between two consecutive distance in $T$. In this case it is easy to see that $o_{f,\sigma,b}(f')$ is uniform on $[n]$.

When $(f'-f)/(n/B)$ is even, it must hold that $r-b\leq s-1$ and thus $n/B\leq 2^{s-1}$. The support of $o_{f,\sigma,b}(f')$ is 
\[
\bigcup_{\text{odd }\ell} \left[2^s\ell - \frac{n}{2B}, 2^s\ell+\frac{n}{2B}\right)
\]
which leaves a gap of width at least $2n/B$ in the middle between two consecutive points in $T$.

When $(f'-f)/(n/B)$ is odd, it must hold that $r-b = s$ and thus $n/B = 2^s$. The support of $o_{f,\sigma,b}(f')$ therefore leaves a gap of width at least $n/B$ in the middle between two consecutive points in $T$. It is easy to see that $o_{f,\sigma,b}(f')$ is uniform on its support.
\end{proof}

The next theorem, which bounds the moment generating function of $\wh{G}_{o_{f}(f')}$, is a straightforward corollary of Lemma~\ref{lem:dist_o_f(f')}.
\begin{lemma}\label{lem:mgf_bound}
Let $n$, $\sigma$ and $b$ be as in Lemma~\ref{lem:dist_o_f(f')}. When $f\neq f'$, $\E \exp(\lambda \wh{G}_{o_{f,\sigma,b}(f')}) \leq M(\lambda)$.
\end{lemma}
\begin{proof}
When $(n/B)\nmid (f-f')$,
\[
\E e^{\lambda \wh{G}_{o_{f,\sigma,b}(f')}} \leq \left(\frac{2}{B}+\frac{1}{n}\right)e^\lambda + \left(1-\frac{2}{B}-\frac{1}{n}\right)e^{\lambda\epsilon} = e^{\lambda\eps}\left[\left(\frac{2}{B}+\frac{1}{n}\right)(e^{\lambda(1-\epsilon)}-1) + 1\right],
\]
where the inequality follows from the fact that $\wh{G}$ is at most $1$ on $[-n/B, n/B]$ as at most $\epsilon$ elsewhere (recall Definition~\ref{def:filter_G}), and the equality from rearranging the terms.

When $f'-f \equiv k(n/B)\pmod{n}$ for even $k$,
\[
\E e^{\lambda \wh{G}_{o_{f,\sigma,b}(f')}} \leq e^{\lambda\epsilon},
\]
since the filter $\wh{G}$ is at most $\epsilon$ outside of $[-n/B,n/B]$ and the distribution $o_{f,\sigma,b}(f')$ is not supported on that interval by Lemma \ref{lem:dist_o_f(f')}.

When $f'-f \equiv k(n/B)\pmod{n}$ for odd $k$,
\[
\E e^{\lambda \wh{G}_{o_{f,\sigma,b}(f')}} \leq \left(\frac{2}{B}+\frac{1}{n}\right)e^\lambda + \left(1-\frac{2}{B}-\frac{1}{n}\right)e^{\lambda\epsilon} = e^{\lambda\eps}\left[\left(\frac{2}{B}+\frac{1}{n}\right)(e^{\lambda(1-\epsilon)}-1) + 1\right],
\]
where the inequality follows again by combining Lemma~\ref{lem:dist_o_f(f')}(iii) and the bounds on $\wh{G}$ from Definition~\ref{def:filter_G}, and the equality is just a rearrangement of terms.
\end{proof}

\subsection{Putting the Pieces Together}\label{subsec:finishing}

We are now ready to verify \eqref{eqn:pessimistic}--\eqref{eqn:pessimistic_step}.
\begin{lemma}[Pessimistic Estimation]
\label{lem:pessimistic}
It holds that
\[
h_r(f,f';\sigma_1,b_1,\dots,\sigma_r,b_r) \geq \Pr\left(A_{f,f'} | \sigma_1,b_1,\dots,\sigma_r,b_r\right).
\]
\end{lemma}
\begin{proof}
Let $z = \sum_{\ell=1}^r G_{o_{f,\ell}(f')}$. Then
\begin{align*}
\Pr\left(A_{f,f'} | \sigma_1,b_1,\dots,\sigma_r,b_r\right) &= \Pr\left( z + \sum_{\ell=r+1}^d G_{o_{f,\ell}(f')} > \beta\right) \\
&= \Pr\left( \exp\left\{\lambda \left(z + \sum_{\ell=r+1}^d G_{o_{f,\ell}(f')}\right)\right\} > e^{\lambda\beta} \right) \\
&\leq e^{-\lambda\beta} e^{\lambda z} \E \exp\left(\lambda\sum_{\ell=r+1}^d G_{o_{f,\ell}(f')}\right) \\
&= e^{-\lambda\beta} e^{\lambda z} (\E \exp(\lambda G_{o_{f,\sigma,b}(f')}))^{d-r}\\
&\leq e^{-\lambda\beta} e^{\lambda z} (M(\lambda))^{d-r},
\end{align*}
where the last inequality follows from Lemma~\ref{lem:mgf_bound}.
\end{proof}

\begin{lemma}[Initial constraint]\label{lem:initial_constraint} It holds that
\[
\sum_{f\neq f'} h_0(f,f') < 1.
\]
\end{lemma}
\begin{proof}
It follows from Lemma~\ref{lem:mgf_bound} that
\begin{align*}
(M(\lambda))^d  &\leq \exp\left\{d \left(\lambda \epsilon + \ln\left(1+\frac{3}{B}(e^{\lambda(1-\epsilon)}-1)\right)  \right)\right\}\\
&\leq \exp\left\{d \left(\lambda\epsilon + \frac{3}{B}(e^{\lambda(1-\epsilon)}-1)\right)\right\}\\
&\leq \exp\left\{d\lambda\left(\epsilon + \frac{3}{B}(1-\epsilon)\right)\right\} \\
&\leq \exp(3d\lambda).
\end{align*}

Recall that we choose $B=\Theta(k)$ and $d=O(k\log n)$. It follows that
\begin{align*}
\sum_{f\neq f'} h_0(f,f') &= e^{-\lambda\beta} \sum_{f\neq f'} (M(\lambda))^d \\
&\leq n^2 \exp\left\{ -C\frac{d}{B} + 3d\lambda \right\}\\
&\leq n^2 \exp(-cd/B) \qquad (\text{by choosing }\lambda = c''/B\text{ for }c''\text{ small enough})\\
&< 1.\qedhere
\end{align*}
\end{proof}
%

\begin{lemma}[Derandomization step]\label{lem:derandomization_step}
It holds that
\[
h_r(f,f';\sigma_1,b_1,\dots,\sigma_r,b_r) \geq \E_{\sigma_{r+1},b_{r+1}}h_{r+1}(f,f';\sigma_1,b_1,\dots,\sigma_r,b_r,\sigma_{r+1},b_{r+1})
\]
\end{lemma}
\begin{proof}
Let $z = \sum_{\ell=1}^r G_{o_{f,r}(f')}^{(\ell)}$. The proposition is equivalent to
\[
\exp\left(\lambda z\right) (M(\lambda))^{d-r} \geq \E_{\sigma,q} \exp\left(\lambda\left(z + G_{o_{f,r}(f')}\right)\right) (M(\lambda))^{d-r-1}, 
\]
This clearly holds by Lemma~\ref{lem:mgf_bound}.
\end{proof}

%% file: sublinear_algo.tex
\section{Sublinear-Time Algorithm}\label{sec:sublinear_time}

In this section, we take the pseudorandom hashings $\{H_r\}_{r\in [d]}$ to be as in Lemma~\ref{lem:G_constraint} and assume that \eqref{eqn:G_constraint_1} holds.

The first lemma concerns $1$-sparse recovery, because, as in earlier works, we shall create $k$ subsignals using hashing, most of which are $1$-sparse.
\begin{lemma}\label{lem:1-sparse}
Suppose that $n$ is a power of $2$. Let $Q = \{0,1,2,4,\dots,n/2\}\subseteq [n]$. Then the following holds: 
Let $x\in \mathbb{C}^n$ and suppose that $|\wh{x}_f| \geq 3\|\wh{x}_{[n]\setminus \{f\}}\|$ for some $f\in [n]$. Then one can recover the frequency $f$ from the samples $x_Q$ in $O(\log n)$ time.
\end{lemma}
\begin{proof}
Define
$
	\theta_{f'} = \left( \frac{2\pi}{n} f' \right) \bmod{2\pi}$. Observe that 
\[	
x_{q} = \frac{1}{\sqrt n}\left(\wh{x}_{f} e^{\sqrt{-1} q \theta_{f}} + \sum_{ f' \neq f} \wh{x}_{f'} e^{\sqrt{-1} q \theta_{f'}}\right),\quad q\in [n],
\]
It follows from Proposition~\ref{fact:phase_est} that $|\arg x_q - (\arg x_f + q \theta_{f})|\leq \pi/8$. When $q = 0$, one has $|\arg x_0 - \arg x_f| \leq \pi/8$, and thus $|\arg(x_q/x_0) - q\theta_{f}| \leq \pi/4$. 

Hence,
\[
\theta_{f} \in I_q,\quad \text{where } I_q := \bigcup_{\ell=0}^{q-1} \left[\frac{2\ell \pi + \arg(x_q/x_0)}{q} - \frac{\pi}{4q},\frac{2\ell \pi + \arg(x_q/x_0)}{q} + \frac{\pi}{4q}\right].
\]
Note that $I_q$ is the union of $q$ disjoint intervals of length $\pi/(2q)$. We may view these intervals as arcs on the unit circle, each arc being of length $\pi/(2q)$, and the left endpoints of every two consecutive arcs having distance $2\pi/ q$. 

Define a series of intervals $\{S_r\}$ for $r=0,1,\dots,\log n - 1$ recursively as
\begin{align*}
S_0 &= I_1,\\
S_{r+1} &= S_r \cap I_{2^{r+1}}.
\end{align*}
It is easy to see, via an inductive argument, that $\theta_f \in S_r$ for all $0\leq r\leq \log n-1$, and $|S_r| \leq \frac{\pi}{2^{r+1}}$. In the end, $S_{\log n-1}$ is an interval of length $\pi/(2n)$, which can contain only one $\theta_{f'}$, and thus we can recover $i$.

Each $S_r$ can be computed in $O(1)$ time from $S_{r-1}$ and thus the overall runtime is $O(\log n)$.
\end{proof}

Now we move to develop our sublinear-time algorithm. The following is an immediate corollary of Lemma~\ref{lem:G_constraint}.
\begin{lemma}\label{lem:bucket_noise}
For each $f$, it holds for at least $8d/10$ indices $r\in [d]$ that
\[
\left|\sum_{f'\in [n]\setminus \{f\}} \hat G_{o_{f,r}(f')} \wh{x}_f\right| \leq \frac{10}{(1-\eps)B} \left\|\wh{x}_{[n]\setminus\{f\}}\right\|_1.
\]
\end{lemma}
\begin{proof}
It follows from Lemma~\ref{lem:G_constraint}, Eq.~\eqref{eqn:bucket} and the observation that $G_{o_{f,r}(f)} \in [1-\eps,1]$.
\end{proof}

As before, we choose $B = 10(1-\eps)^{-1}\beta k$ rounded to the closest power of $2$; $\beta$ is some constant to be determined. The following is a lemma for Algorithm~\ref{alg:recovery_sub}, which gives the same guarantees as Lemma~\ref{lem:recovery_of_HH}. 

\begin{lemma}\label{lem:recovery_of_HH_sublinear}
Suppose that $x,\hat z,\nu$ be the input to Algorithm~\ref{alg:recovery_sub}. Let $w = \hat x -\hat z$. When $\nu\geq \frac{16}{\beta k}\|\wh{w}\|_1$, the output $\wh{w}'$ of Algorithm~\ref{alg:recovery_sub} satisfies
\begin{enumerate}[label=(\roman*),noitemsep]
	\item $|\wh{w}_f| \geq (7/16)\nu$ for all $f \in \supp(\wh{w}')$.
	\item $|\wh{w}_f-\wh{w}'_f| \leq |\wh{w}_f|/7$ for all $i \in \supp(\wh{w}')$;
	\item $\supp(\wh{w}')$ contains all $f$ such that $|\wh{w}_f|\geq \nu$;
\end{enumerate}
\end{lemma}
\begin{proof}
The proof of (i) and (ii) are the same as the proof of Lemma~\ref{lem:recovery_of_HH}.
Next we prove (iii). When $|\hat{w}_f|\geq \nu$, we have 
\[
|\wh{G}_{o_{f,r}(i)}\wh{w}_f|\geq (1-\eps)\nu \geq \frac{16(1-\eps)}{\beta k}\|\wh{w}_{[n]\setminus \{f\}}\|_1.
\]
Hence for the signal $y_r\in \mathbb{C}^n$ defined via its Fourier coefficients as
\[
(\wh{y_r})_{f'} = \wh{G}_{o_{f,r}(f')} \wh{x}_{f'},
\]
By Lemma~\ref{lem:bucket_noise}, since $16(1-\eps)\geq 3$, we see that $y_r$ with frequency $f$ satisfies the condition of Lemma~\ref{lem:1-sparse} and thus it will be recovered in at least $8d/10$ repetitions $r\in [d]$. The measurements are exactly $(m_H)_{h(f)}$ with $q\in Q$. The thresholding argument is the same as in the proof of Lemma~\ref{lem:recovery_of_HH}.
\end{proof}

Observe that Lemma~\ref{lem:main_loop_invariant} continues to hold if we replace Algorithm~\ref{alg:linear_recovery} with Algorithm~\ref{alg:recovery_sub} and Lemma~\ref{lem:recovery_of_HH} with Lemma~\ref{lem:recovery_of_HH_sublinear}. Now we are ready to prove our main theorem, Theorem~\ref{thm:intro_main_2}, on the sublinear-time algorithm.

\begin{proof}[Proof of Theorem~\ref{thm:intro_main_2}]
The recovery guarantee follows identically as in the proof of Theorem~\ref{thm:intro_main_1}.

The measurements are $u_q$ for $q\in Q$ in each of the $d$ repetitions, and calculating each $u_q$ requires $O(k)$ measurements (Lemma~\ref{lem:hashtobins}). There measurements are reused throughout the iteration in the overall algorithm, hence there are $O(k d|Q|) = O(k \cdot k\log n \cdot \log n) = O(k^2\log^2 n)$ measurements in total.

Each call to \textsc{SubRecovery} runs in time $O(d(B\log B + \|\hat z\|_0\log n + B\log n) + kd) = O(k^2\log^2 n + k\|\hat z\|_0\log^2 n) = O(k^2 \log^2 n)$, where we use the fact that $\|\hat z\| = O(k)$ from Lemma~\ref{lem:main_loop_invariant}(a). The overall runtime is therefore $O(k^2\log^3 n)$.
\end{proof}

\begin{algorithm}
\caption{Sublinear-time Sparse Recovery for $\wh{x}-\wh{z}$}\label{alg:recovery_sub}
\begin{algorithmic}
\Procedure{SubRecovery}{$x,\wh{z},\nu$}
\State $\Lambda = \emptyset$
\For{$r=1$ to $d$}
	\For{each $q\in Q$}	\Comment{$Q$ as in Lemma \ref{lem:1-sparse}}
		\State $u_q\gets \Call{HashToBins}{x,\wh{z},(\sigma_r,q,b_r)}$
	\EndFor
	\For{$b = 1$ to $B$}
		\State $f\gets \Call{OneSparseRecovery}{\{(u_q)_b\}_{q\in Q}}$
		\State $\Lambda = \Lambda\cup \{f\}$
		\State $v_{f,r} = (u_0)_{h_r(f)}$
	\EndFor
\EndFor
\State $\wh{w}' \gets 0$
\For{each $f\in \Lambda$}	
	\State $v_f \gets \median_r v_{f,r}$ \Comment{median is taken over all $r$ such that $v_{f,r}$ exists}
	\If{$|v_f|\geq \nu/2$}
		\State $\wh{w}'_f\gets v_f$
	\EndIf
\EndFor
\State \Return $\wh{w}'$
\EndProcedure
\end{algorithmic}
\end{algorithm}

%% file: incoherence.tex
\section{Incoherent Matrices via Subsampling DFT Matrix}\label{sec:incoherent_appendix}

Consider an $N\times N$ unitary matrix $A$ and assume that $|A_{i,j}|\leq C/\sqrt{n}$ for all $i,j$. Our goal in this section is to show how to sample deterministically $m=C_m k^2\log n$ rows of $A$, obtaining a matrix $B$, such that $|\langle B_i,B_j\rangle| \leq m/(kn)$ for all pairs $i\neq j$. Once we have such $B$, the rescaled matrix $B'=\sqrt{ \frac{n}{m}}A$ is a $(1/k)$-incoherent matrix, that is,  $|\langle B'_i,B'_j\rangle| \leq 1/k$ for all pairs $i\neq j$.

Let $\delta_1,\dots,\delta_n$ be i.i.d.\@ Bernoulli variables with $\Pr(\delta_\ell = 1) = p$ for some $p = m/n$. Let $i,j\in [n]$ such that $i\neq j$, then
\[
\langle B_i,B_j\rangle = \sum_\ell \delta_\ell A_{\ell,i}\overline{A}_{\ell,j}.
\]

Let $z_\ell = A_{\ell,i}\overline{A}_{\ell,j}$, then $|z_\ell|\leq \eta$, where $\eta=C^2/n$. We consider the real and the imaginary parts separately, since for a complex random variable $Z$,
\[
\Pr(|Z| > t) \leq \Pr\left(|\Re Z| > \frac{t}{\sqrt{2}}\right) + \Pr\left(|\Im Z| > \frac{t}{\sqrt{2}}\right).
\]
Hence it suffices to consider the real variable problem as follows. Suppose that $a_1,\dots,a_n\in \R$ satisfy $|a_i|\leq \eta$, and consider the centred sum $S = \sum_i (\delta_i - p)a_i$. We wish to find $\delta_1,\dots,\delta_n$ deterministically such that $|S| \leq m/(kn)$.

Define the pessimistic estimator to be
\[
f_r(\delta_1,\dots,\delta_r) = e^{-\lambda t}\left(e^{\lambda\sum_{i=1}^r (\delta_i-p)a_i} \prod_{i=r+1}^n M_i(\lambda) + e^{-\lambda \sum_{i=1}^r (\delta_i-p)a_i} \prod_{i=r+1}^n M_i(-\lambda)\right)
\]

The moment generating function of $(\delta_i - p)a_i$ is
\[
M_i(\lambda) = p e^{\lambda (1-p) a_i} + (1-p) e^{-\lambda p a_i},\quad i = 1,\dots,n.
\]

\paragraph*{Pessimistic Estimation} Let $w = \sum_{i=1}^r (\delta_i - p) a_i$, where $\delta_1,\dots,\delta_r$ have been fixed.
\begin{align*}
\Pr(|S| > t | \delta_1,\dots,\delta_r) &= \Pr(S > t | \delta_1,\dots,\delta_r ) + \Pr(-S > t | \delta_1,\dots,\delta_r)\\
&= \Pr(e^{\lambda S} > e^{\lambda t} | \delta_1,\dots,\delta_r) + \Pr(e^{-\lambda S} > e^{\lambda t} | \delta_1,\dots,\delta_r)\\
&\leq e^{-\lambda t} \E (e^{\lambda S} + e^{-\lambda S} | \delta_1,\dots,\delta_r)   \\
&= e^{-\lambda t}\left(e^{\lambda w} \prod_{i=r+1}^n M_i(\lambda) + e^{-\lambda w} \prod_{i=r+1}^n M_i(-\lambda)\right)\\
&= f_r(\delta_1,\dots,\delta_r).
\end{align*}

\paragraph*{Derandomization step} One can show first that
\begin{equation}\label{eqn:incoherent_derandom_step_critical}
f_r(\delta_1,\dots,\delta_r) = p f_{r+1}(\delta_1,\dots,\delta_r,1) + (1-p) f_{r+1}(\delta_1,\dots,\delta_r,0),
\end{equation}
which is equivalent to 
\begin{equation}\label{eqn:incoherent_derandom_aux}
e^{\lambda w}M_{r+1}(\lambda)\prod_{i=r+2}^n M_i(\lambda) + e^{-\lambda w}M_{r+1}(-\lambda)\prod_{i=r+2}^n M_i(-\lambda) = p M' + (1-p) M'',
\end{equation}
where
\begin{align*}
M' &=  e^{\lambda (w+(1-p)a_i)}\prod_{i=r+2}^n M_i(\lambda) + e^{-\lambda (w+(1-p)a_i)}\prod_{i=r+2}^n M_i(-\lambda),\\
M'' &=  e^{\lambda (w-p a_i)}\prod_{i=r+2}^n M_i(\lambda) + e^{-\lambda (w -p a_i)}\prod_{i=r+2}^n M_i(-\lambda).
\end{align*}
It is now clear that the left-hand side of \eqref{eqn:incoherent_derandom_aux} is $pM' + (1-p)M''$, and therefore \eqref{eqn:incoherent_derandom_step_critical} holds. This implies that
\[
f_r(\delta_1,\dots,\delta_r) \geq \min\{ f_{r+1}(\delta_1,\dots,\delta_r,1), f_{r+1}(\delta_1,\dots,\delta_r,0)\}.
\]

\paragraph*{Initial condition} This is a standard argument for Bernstein's inequality. For notational convenience, let $\phi(x) = (e^{\lambda x}-\lambda x-1)/x^2$. Note that $\phi(x)$ is increasing on $(0,\infty)$. Using Taylor's expansion, one can bound that (see~\cite[p35]{BLM})
\[
M_i(\lambda) \leq \exp\left(\phi(|a_i|) p(1-p) a_i^2\right)\leq \exp\left(\phi(\eta) p(1-p) a_i^2\right).
\]
and (see~\cite[p98]{tropp:book})
\[
\phi(\eta) \leq \frac{\lambda^2/2}{1-\lambda \eta/3},\quad \lambda < \frac{3}{\eta}.
\]
It then follows (see~\cite[p98]{tropp:book}) that
\begin{align*}
\Pr(|S| > t) &\leq 2e^{-\lambda t} e^{\phi(\eta) p(1-p) \sum_i |a_i|^2} \\
&\leq 2\exp\left(-\lambda t + n\eta^2 p(1-p) \frac{\lambda^2/2}{1-\lambda \eta/3}\right)\\ &\leq 2\exp\left(-\frac{t^2/2}{n\eta^2 p(1-p) + t\eta/3}\right),
\end{align*}
provided that $\lambda = t/(n \eta^2 p(1-p) + t\eta/3) \in (0,3/\eta)$.

When $t = m/(kn)$, $p = m/n$ and $\eta=C^2/n$, $\lambda \simeq \log n < 3/\eta$ and the above probability is at most
\[
2\exp\left(-\frac{1}{2(C^4+\frac{c}{3})}\cdot \frac{m}{k^2}\right) \leq 2\exp(-c' C_m\log n) \leq \frac{1}{n^3},
\]
provided that $C_m$ is large enough.

Therefore at step $r$, the algorithm minimizes $f_{r+1}(\delta_1,\dots,\delta_{r+1})$ by choosing $\delta_{r+1}$, and at the end of step $r+1$, all $\delta_1,\dots,\delta_r$ have been fixed and such that $|\sum_i (\delta_i-p)a_i| \leq t$.

Now we return to the original incoherence problem in the complex case. We can define $2n(n-1)$ events, $E_{i,j}$ and $F_{i,j}$, for every pair $i\neq j$ as
\[
E_{i,j} = \left\{|\Re\langle B_i,B_j\rangle| > t\right\},\quad F_{i,j} = \left\{|\Im\langle B_i,B_j\rangle| > t\right\}
\]
For each pair of $i\neq j$, using the preceding argument, we have pessimistic estimators $f^{(1)}_r(i,j;\delta_1,\dots,\delta_r)$ by setting $a_\ell = \Re B_{i,\ell}\overline{B_{\ell,i}}$ and $f^{(2)}_r(i,j;\delta_1,\dots,\delta_r)$ by setting $a_\ell = \Im B_{i,\ell}\overline{B_{\ell,j}}$ such that
\begin{itemize}
	\item (pessimistic estimation) 
	\begin{gather*}
			f^{(1)}_r(i,j;\delta_1,\dots,\delta_r)\geq \Pr(E_{i,j}|\delta_1,\dots,\delta_r)\\
			f^{(2)}_r(i,j;\delta_1,\dots,\delta_r)\geq \Pr(F_{i,j}|\delta_1,\dots,\delta_r)
	\end{gather*}
	\item (derandomization step)
	\begin{equation}\label{eqn:derandom_E_F}
		f^{(s)}_r(i,j;\delta_1,\dots,\delta_r) = p f^{(s)}_{r+1}(i,j;\delta_1,\dots,\delta_r,1) + (1-p) f^{(s)}_{r+1}(i,j;\delta_1,\dots,\delta_r,0), \quad s = 1,2
	\end{equation}
	\item (initial condition)
	\[
		\sum_{i\neq j} f_0^{(1)}(i,j) + f_0^{(2)}(i,j) < \frac{1}{n}.
	\]
\end{itemize}
Note that \eqref{eqn:derandom_E_F} implies
\begin{multline*}
\sum_{i\neq j} \left[f^{(1)}_r(i,j;\delta_1,\dots,\delta_r) + f^{(2)}_r(i,j;\delta_1,\dots,\delta_r)\right] \\
\geq \min_{\delta_{r+1}\in \{0,1\}} \sum_{i\neq j} \left[f^{(1)}_r(i,j;\delta_1,\dots,\delta_r,\delta_{r+1}) + f^{(2)}_r(i,j;\delta_1,\dots,\delta_r,\delta_{r+1})\right].
\end{multline*}

In addition, we also need to control the number of $\delta_i$'s which take value $1$; we want this number to be $O(m)$. This can be achieved by combining another derandomization procedure on $\sum_i \delta_i$ using one-sided Chernoff bounds. Define the event $G = \{\sum_i \delta_i > 2m\}$. Then for $\kappa > 0$,
\begin{align*}
\Pr(G|\delta_1,\dots,\delta_r) &= \Pr(e^{\kappa\sum_i \delta_i} > e^{2m\kappa} | \delta_1,\dots,\delta_r) \\
&\leq \exp\left(-2m\kappa + \kappa\sum_{i=1}^r \delta_i\right) \prod_{i=r+1}^n \E e^{\kappa \delta_i} \\
&= \exp\left(-2m\kappa + \kappa\sum_{i=1}^r \delta_i\right) (M(\kappa))^{n-r},
\end{align*}
where
\[
M(\kappa) = \E e^{\kappa \delta_i} = p e^{\kappa} + 1 - p
\]
is the moment generating function of $\delta_i$. Define our pessimistic estimator to be
\[
g_r(\delta_1,\dots,\delta_r) = \exp\left(-2m\kappa + \kappa\sum_{i=1}^r \delta_i\right) (M(\kappa))^{n-r},
\]
then, similar to the proof in Section~\ref{sec:linear_time}, we have
\begin{itemize}
	\item (pessimistic estimation) 
	\[			
		g(\delta_1,\dots,\delta_r)\geq \Pr(G|\delta_1,\dots,\delta_r),
	\]
	\item (derandomization step)
	\[
		g_r(\delta_1,\dots,\delta_r) \geq p g_{r+1}(\delta_1,\dots,\delta_r,1) + (1-p) g_{r+1}(\delta_1,\dots,\delta_r,0),
	\]
	\item (initial condition) When $\kappa$ is small enough and $C_m$ large enough,
	\[
		g_0 < \frac{1}{2}.
	\]
\end{itemize}

Overall, our standard derandomization procedure, which at step $r$ chooses $\delta_{r+1}\in \{0,1\}$ that minimizes
\[
\sum_{i\neq j} \left[ f^{(1)}_{r+1}(i,j;\delta_1,\dots,\delta_r,\delta_{r+1}) + f^{(2)}_{r+1}(i,j;\delta_1,\dots,\delta_r,\delta_{r+1})\right] + g_{r+1}(\delta_1,\dots,\delta_r,\delta_{r+1})
\]
will find $\delta_1,\dots,\delta_r$ such that none of $E_{i,j}$ and $F_{i,j}$ and $G$ holds, which implies that $|\langle B_i,B_j\rangle|\leq t = m/(kn)$ for all $i\neq j$ and $\sum \delta_i \leq 2m$. That is, we have chosen $2m$ rows of $A$, obtaining a matrix $B$ of incoherence at most $m/(kn)$.

%% file: weil.tex
\section{Incoherent Matrices and Analytic Number Theory}\label{sec:weil}

In this section we give new results via the connection between the incoherent matrices and the exponential sum of characters, a classical quantity of interest in analytic number theory. Such connection has been formerly exploited, for instance, by Xu~\cite{Xu_RIP} and Bourgain et al.~\cite{bourgain2011breaking} for explicit constructions of RIP matrices. 
We utilize the connection bidirectionally: we shall give explicit constructions of incoherent matrices using exponential sums, and improve the lower bound of an exponential sum using a lower bound of incoherent matrices. 

\subsection{A simple construction via Gauss sums}

We give a rather simple construction of an $(1/\sqrt n)$-incoherent matrix $M \in \mathbb{C}^{\frac{n+1}{2} \times n}$. It is expected that Gauss sums will behave nicely for incoherent matrices, since they have the optimal rate of cancellation: summing $p/2$ elements gives cancellation $\sqrt{p}$. Let $p$ be a prime number and let $Q = \{ x \in \Z_p: \exists y \in \Z_p, y^2 = x\}$, i.e.\ the set of quadratic residues in $\Z_p$, including $0$. It is a standard fact that $|Q| = (p+1)/2$. We shall show that the rows of the DFT matrix indexed by the elements of $Q$ give an incoherent matrix with an appropriate scaling.
Let $\omega = e^{2\pi \sqrt{-1}/p}$. Observe that 
	\[
	\left|\sum_{x \in Q} \omega^{ t x } \right| = \left|\frac{1}{2} + \frac{1}{2}\sum_{x \in \Z_p} \omega^{t x^2 } \right| \leq \frac{1}{2} + \sqrt{p},	\quad \forall t \in \Z_p^*,
	\]
where the last inequality follows from the triangle inequality and the standard property of Gauss sums (see, e.g., \cite[p91]{ireland1990classical}). 

Now, let $M\in \C^{|Q|\times p}$ be defined as $M_{x,t} = \omega^{tx}$ for $x\in Q$ and $t\in \Z_p$. For every pair $(t_1,t_2) \in \Z_p \times \Z_p$ with $t_1 \neq t_2$, we have that the inner product of the $t_1$-th and the $t_2$-th column of $M$ is exactly $\sum_{x \in Q} \omega^{(t_1-t_2) x}$. Normalising $M$ gives the desired result.

\subsection{Proof of Theorem \ref{thm:fourier_strongly_explicit}}

In the previous subsection we obtained an incoherent matrix by picking the rows of DFT indexed by quadratic residues, i.e. quadratic polynomials. Motivated by this, we show that taking polynomials of a higher degree can give an improved result that works in a larger range of parameters. We shall need the following deep theorem of Weyl.

\begin{theorem}[{\cite[Theorem 4.3]{nathanson1996additive}}]\label{thm:weyl}
Let $M,N,q$ be positive integers and $\alpha$ an integer such that $(\alpha,q ) =1$. If $g$ is a real polynomial of degree $d \geq 2$ with leading coefficient $a$ such that $|a - \frac{\alpha}{q}| \leq q^{-2}$, then for any $\epsilon >0$ we have 
\[	
\left| \sum_{x=M+1}^{M+N} e^{2\pi\I g(x) } \right| = O\left( N^{1+\epsilon} \left( \frac{1}{q} + \frac{1}{N} + \frac{q}{N^d}\right)^{2^{1-d}} \right),
\]	  
where the hidden constant in the $O$-notation depends on $d$ and $\eps$.
\end{theorem}

We are now ready to prove Theorem \ref{thm:fourier_strongly_explicit}.
\begin{proof}
Pick any polynomial $g$ of degree $d$ such that every coefficient of $g$ is an integer multiple of $1/p$. Pick also any $m$ consecutive points in $\Z_p$; we can just take $0$ to $m-1$. Take the rows of DFT indexed by $g$ evaluated on these $m$ consecutive points. We shall show that after appropriate normalization this corresponds to an incoherent matrix of the desired form. The inner product between two columns indexed by $t_1,t_2$ of the formed matrix is 
	\[ \sum_{x=0}^{m-1} e^{2\pi \I g(t_1)} \cdot e^{-2\pi \I g(t_2)} = \sum_{x=0}^{m-1} e^{2\pi \I \left( g(t_1) - g(t_2) \right)}.
	\]
Observe that $g(t_1) - g(t_2)$ is a $d$-degree polynomial where every coefficient is an integer multiple of $1/p$. Applying Theorem \ref{thm:weyl} with $N = m$, $a = \frac{t}{p}$, $\alpha = t$, $q = p$ and noticing that $q\geq N$, we see that the above sum is at most
\[	
O\left(m^{1+\epsilon} 	\left( \frac{1}{m} + \frac{p}{m^d} \right)^{2^{1-d}}\right).
\]
Rescale the formed matrix by $1/\sqrt{m}$, the incoherence of the matrix is rescaled by $1/m$ and thus becomes
\[
O\left(m^{\epsilon} \left( \frac{1}{m} + \frac{p}{m^d} \right)^{2^{1-d}}\right),
\]
yielding the desired result.
\end{proof}

\subsection{Proof of Theorem~\ref{thm:fourier_incoherent_subgroup}}

\begin{proof}
Suppose that $g$ is a generator of the multiplicative cyclic group $\Z_p^\ast$ (we shall show how to find such $g$ later). For every $d$ that divides $p-1$ we shall take the rows of DFT indexed by the multiplicative subgroup $G$ that is generated by $g^{(p-1)/d}$. Since $g$ is a generator of $\Z_p^\ast$ it must hold that $|G| = d$. The incoherence bound follows by a classical fact that (see, e.g.~\cite{kurlberg}) for any $t_1,t_2 \in \mathbb{Z}_p$ with $t_1\neq t_2$,
\[	\left | \sum_{a \in G} e^{2\pi\I \frac{a(t_1-t_2)}{p}} \right| \leq \sqrt{p}.
\]
Rescaling gives the desired incoherence bound.

To find a generator $g$ of $\Z_p^\ast$ is a classic problem with a rich research history. We include a simple, standard algorithm below for completeness.

The first step is to factor $p-1$ in $\tilde{O}(\sqrt{p})$ time. We can find all primes smaller than $\sqrt{p-1}$ in $O(\sqrt{p} \log \log p)$ time using Eratosthene's sieve. For each such prime $q$ we shall find the highest power $q^\ell$ which divides $p-1$. Let $t$ be the number that is obtained after dividing $p-1$ with $q^\ell$ for all such $q,\ell$. If $t\neq 1$, it must be a prime, otherwise for $t=ab$ one of $a,b$ would be at most $\sqrt{t} \leq \sqrt{p-1}$. 

Now we are ready to find a generator $g$. It is known that the smallest generator of $\Z_p^\ast$ is $O(p^{1/4+\epsilon})$~\cite[Theorem 3]{burgess} and thus we shall iterate over the first $O(p^{1/4+\eps})$ elements of $\Z_p^\ast$ and check if every such element $z$ is a generator by checking whether $z^{(p-1)/d}\neq 1$ in $\mathbb{Z}_p^*$ for all prime divisors $d$ of $p-1$. To ensure that such a $z$ is a generator, observe first that the checking condition guarantees that $z$ is of order $p-1$, and checking only prime $d$ suffices (since if $d$ is composite and $z^{(p-1)/d} =1$ this implies $z^{(p-1)/{d'}}=1$ for all divisors $d'$ of $d$); moreover, it is a basic fact in group theory that the order of any subgroup divides the order of the group and hence we need only look at divisors of $p-1$. The runtime of this part is $\widetilde{O}(p^{1/4+\eps})$.
\end{proof}


\subsection{Strengthening the lower bound in \cite{winterhof2001} }

The lower bound in \cite{winterhof2001} states that for any $n,d \geq 2$ and with $\mathrm{gcd}(d,n-1)=1$, any subset $S \subseteq \Z_n$, there exists $b \in \Z_n^\ast$ and an irreducible $d$-degree polynomial $g$ with coefficients in $\Z_n$, such that

	\[	\left | \sum_{x \in S} e^{2\pi \sqrt{-1} \frac{bg(x)}{n} } \right| = \Omega\left(\sqrt{|S|} \right).	\]

With the connection to incoherent matrices and the lower bound of Alon, we obtain a much stronger result. In fact we have for any $d \geq 1$ and any polynomial $g$ with coefficients in $\Z_n$ that
\begin{equation}\label{eqn:sum_lb}
\left | \sum_{x \in S} e^{2\pi \sqrt{-1} \frac{bg(x)}{n} } \right| = \Omega\left(\sqrt{|S|\log_{|S|}n} \right)
\end{equation}
for some $b\in \Z_n^\ast$, provided that $|S| = \Omega(\log n/\log \log n)$. In the case that $|S| = O( \log n/ \log \log n)$ we still have a lower bound of $\Omega( \sqrt{|S|} )$.

Note that the condition $d \geq 2$ has been relaxed to $d \geq 1$, the assumption that $\gcd(d,n-1)=1$ has been removed, the conclusion ``there exists an irreducible polynomial'' has been replaced with the condition ``for any polynomial'', and the right-hand side has been amplified by a multiplicative factor of $\sqrt{\log_{|S|} n}$ for $|S| = \Omega(\log n/ \log \log n)$. 

Our new lower bound follows immediately from Alon's lower bound on incoherent matrices~\cite{alon2009perturbed}. Indeed, assume that there exists a polynomial $g$ such that for all $b\in \Z_n^\ast$ the left-hand side of~\eqref{eqn:sum_lb} is at most $c\sqrt{|S| \log_{|S|}n}$ for some absolute constant $c$. Consider the matrix with the rows of the DFT matrix indexed by numbers $\{g(x): x \in S\}$ (some rows of the DFT matrix may appear more than once). Observe that after normalizing the matrix by $\frac{1}{\sqrt{|S|}}$, the incoherence is 
	\[	
	\frac{1}{|S|} \sup_{\substack{(t_1,t_2) \in \Z_n \times \Z_n\\ t_1 \neq t_2}} \left| \sum_{x \in S } e^{2\pi \sqrt{-1} \frac{(t_1-t_2)g(x)}{n} } \right| \leq \frac{c \sqrt{\log_{|S|}n}}{ \sqrt{|S|} } = c\sqrt{\frac{\log n}{|S|\log |S|}}.
	\]
This would violate the lower bound in \cite{alon2009perturbed}, which states that an $m\times n$ $(1/k)$-incoherent matrix must satisfy $m \geq \alpha \cdot k^2 \log_k n$ for some absolute constant $\alpha$, since
\[
|S| <\alpha \left( \frac{|S|\log |S|}{ c^2 \log n} \right) \cdot \frac{\log n}{\frac{1}{2} \log \left( \frac{|S|\log|S|}{ c^2\log n} \right)}
\]
for $c$ small enough, when $|S| = \Omega(\log n/\log\log n)$. In the case of $|S| = O( \log n / \log \log n)$ we can still use the quadratic bound ($m = \Omega(k^2)$) on incoherent matrices to obtain a bound of $\Omega( \sqrt{|S|} )$. 

%% file: open_problems.tex
\section{Open Problems and Future Direction}
A direction of research is to design deterministic schemes that break the quadratic barrier for signals with structured Fourier support. For example, subsampling the rows of the DFT matrix to obtain \textsc{RIP} matrices depends highly on the structure of the vectors we would like to preserve. The more additive structure the support of a $k$-sparse vector $x$ has, the worse is the concentration of a random Fourier coefficient of $x$. Equivalently, the less additive structure the support of $x$ has, the flatter its Fourier transform is, and hence, the better concentration bounds we obtain. The concentration in the extreme case, when the support of $x$ is ``dissociated'', is captured by the renowned Rudin's inequality in additive combinatorics (see, e.g.~\cite[Lemma 4.33]{tao2006additive}). We thus believe that it is an interesting direction to use  machinery from the field of additive combinatorics and the relevant fields in order to obtain new constructions and algorithms, at least for interesting subclasses of structured signals.

\section{Acknowledgements} 

We would like to thank anonymous reviewers for their valuable feedback.

%% file: ell_inf_reduction.tex
\section{Reduction of the $\ell_\infty$ norm}\label{sec:omitted_proofs}

\begin{replemma}{lem:recovery_of_HH}
Suppose that $x,\hat z,\nu$ be the input to Algorithm~\ref{alg:recovery_sub}. Let $w = \hat x -\hat z$. When $\nu\geq \frac{16}{\beta k}\|\wh{w}\|_1$, the output $\wh{w}'$ of Algorithm~\ref{alg:recovery_sub} satisfies
\begin{enumerate}[label=(\roman*),noitemsep]
	\item $|\wh{w}_f| \geq (7/16)\nu$ for all $i \in \supp(\wh{w}')$.
	\item $|\wh{w}_f-\wh{w}'_f| \leq |\wh{w}_f|/7$ for all $i \in \supp(\wh{w}')$;
	\item $\supp(\wh{w}')$ contains all $i$ such that $|\wh{w}_f|\geq \nu$;
\end{enumerate}
\end{replemma}
\begin{proof}
By the recovery guarantee we know that
\[
|\wh{w}_f-\wh{w}'_f| \leq \frac{\|w\|_1}{\beta k} \leq \frac{\nu}{16}.
\]
	By thresholding, it must hold for $i\in\supp(\wh{w}')$ that $|\wh{w}'_f|\geq \nu/2$ and thus
	\[
		|\wh{w}_f| \geq \frac{\nu}{2} - \frac{\nu}{16} = \frac{7}{16}\nu,
	\]
	which proves (i). Thus 
	\[
 	|\wh{w}_f-\wh{w}'_f| \leq \frac{\nu}{16} \leq \frac{1}{7}|\wh{w}_f|,
	\]
	which proves (ii). Next we prove (iii). When $|\hat{w}_f|\geq \nu$, we have 
\[
|\wh{G}_{o_{f,r}(f)}\wh{w}_f|\geq (1-\eps)\nu \geq \frac{16(1-\eps)}{\beta k}\|\wh{w}_{[n]\setminus \{f\}}\|_1.
\]
Hence for the signal $y_r\in \mathbb{C}^n$ defined via its Fourier coefficients as
\[
(\wh{y_r})_{f'} = \wh{G}_{o_{i,r}(j)} \wh{x}_{f'},
\]
By Lemma~\ref{lem:bucket_noise}, since $16(1-\eps)\geq 3$, we see that $y_r$ with index $i$ satisfies the condition of Lemma~\ref{lem:1-sparse} and thus it will be recovered in at least $8d/10$ indices $r\in [d]$. The measurements are exactly $(m_H)_{h(i)}$ with $q\in Q$. The recovered estimate is at least $\nu - \nu/16 > \nu/2$ and thus  the median estimate will pass the thresholding, and $i\in \supp(\wh{w}')$.
\end{proof}

Let $H = H(x,k)$ and $I = \{f: |\wh{x}_f|\geq \frac{1}{\rho k}\|x_{-k}\|_1 \}$. By the SNR assumption of $\wh{x}$, we have that $\|\wh{x}_{H}\|_1\leq k\|\wh{x}\|_\infty \leq R^\ast\|\wh{x}_{-k}\|_1$ and thus $\|\wh{x}\|_1\leq (R^\ast+1)\|\wh{x}_{-k}\|_1$. Let $r^{(t)}$ be the residual vector at the beginning of the $t$-th step in the iteration. The threshold in the $t$-th step is
\[
\nu^{(t)} = C\mu\gamma^{T-t},
\]
where $C \geq 1, \gamma > 1$ are constants to be determined.

\begin{replemma}{lem:main_loop_invariant}
There exist $C,\beta,\rho,\gamma$ such that it holds for all $0\leq t \leq T$ that
\begin{enumerate}[label=(\alph*),noitemsep]
	\item $\wh{x}_f = r^{(t)}_f$ for all $f\notin I$;
	\item $|r^{(t)}_f| \leq |\wh{x}_f|$ for all $f$.
	\item $\|r^{(t)}_f\|_\infty \leq \nu^{(t)}$;
\end{enumerate}
\end{replemma}

\begin{proof}
We prove the three properties inductively. The base case is $t=0$, where all properties clearly hold, noticing that $\mu \gamma^T = \|x\|_\infty$.

Next we prove the inductive step from $t$ to $t+1$. Note that
\begin{align*}
\|r^{(t)}\|_1 &\leq \|r^{(t)}_H\|_1 + \|r^{(t)}_{H^c}\|_1 \\
&= \|r^{(t)}_{H\cap I}\|_1 + \|r^{(t)}_{H\setminus I}\|_1 + \|x_{-k}\|_1 \\
&\leq k\cdot \|r_I\|_\infty + k\cdot \frac{1}{\rho k}\|x_{-k}\|_1 + \|x_{-k}\|_1 \\
&\leq k\cdot C\mu\gamma^{T-t} + \left(1 + \frac{1}{\rho}\right)\|x_{-k}\|_1 \\
&= C\gamma^{T-t}\|x_{-k}\|_1 + \left(1 + \frac{1}{\rho}\right)\|x_{-k}\|_1
\end{align*}
When
\begin{equation}\label{eqn:cons1}
	C \left(1-\frac{16}{\rho}\right) \geq \frac{16}{\beta}\left(1+\frac{1}{\rho}\right),
\end{equation}
it holds that
\[
\nu^{(t)} \geq \frac{16}{\beta k}\|r^{(t)}\|_1
\]
and thus Lemma~\ref{lem:recovery_of_HH} applies.

From Lemma~\ref{lem:recovery_of_HH}(i), we know that when
\begin{equation}\label{eqn:cons2}
	\frac{7}{16}C \geq \frac{1}{\rho},
\end{equation}
no coordinates in $I^c$ will be modified. This proves (a).

Lemma~\ref{lem:recovery_of_HH}(ii) implies (b).

To prove (c), let $J = \{f\in I: |r^{(t)}_f| \geq \nu^{(t+1)}\}$. By Lemma~\ref{lem:recovery_of_HH}(iii), all coordinates in $J$ will be recovered. Hence for $f\in J$,
\[
|r^{(t+1)}_f| \leq \frac{1}{7}|r^{(t)}_f|\leq \frac{1}{7}\nu^{(t)} \leq \nu^{(t+1)},
\]
provided that
\begin{equation}\label{eqn:cons3}
\frac{1}{7}\leq \frac{1}{\gamma}.
\end{equation}
For $f\in I\setminus J$, the definition of $J$ implies that $|r^{(t+1)}_f| \leq \nu^{(t+1)}$. This proves (c).

We can take $C = 2$, $\rho = 32$, $\beta = 32$, $\gamma = 2$, which satisfy all the constraints \eqref{eqn:cons1}, \eqref{eqn:cons2} and \eqref{eqn:cons3}.
\end{proof}